\documentclass[journal,12pt,onecolumn,draftclsnofoot,]{IEEEtran}

\usepackage{mathtools,dsfont}
\usepackage{psfrag,bbm,graphicx}
\usepackage{comment,balance}
\usepackage{epstopdf}
\usepackage{epsfig,subfig}
\usepackage{multicol}
\usepackage{multirow, cite}
\usepackage{amsfonts, color}
\usepackage{array,cuted,lipsum}
\usepackage{amssymb,amsmath,amsthm}
\usepackage[linesnumbered,ruled,vlined]{algorithm2e}
\usepackage{caption}
\allowdisplaybreaks

\newtheorem{assumption}{Assumption}

\newcommand{\PP}{\mathbb{P}}
\newcommand{\E}{\mathbb{E}}

\newcommand{\Indi}{\mathds{1}}
\newcommand{\nn}{\nonumber}
\newcommand{\muB}{\boldsymbol{\mu}}

\newcommand{\calP}{\mathcal{P}}

\newcommand{\OO}{\text{O}}
\newcommand{\bfK}{\textbf{K}(T)}

\newtheorem{lemma}{Lemma}

\newtheorem{theorem}{Theorem}
\newtheorem{definition}{Definition}

\title{Regret of Age-of-Information Bandits}
 \author{Santosh Fatale, Kavya Bhandari, Urvidh Narula, Sharayu Moharir,\\ and Manjesh Kumar Hanawal \\
 Indian Institute of Technology Bombay}
\newfont{\mycrnotice}{ptmr8t at 7pt}
\newfont{\myconfname}{ptmri8t at 7pt}

\begin{document}
\maketitle

\begin{abstract}
	We consider a system with a single source that measures/tracks a time-varying quantity and periodically attempts to report these measurements to a monitoring station. Each update from the source has to be scheduled on one of $K$ available communication channels. The probability of success of each attempted communication is a function of the channel used. This function is unknown to the scheduler. 
	
	The metric of interest is the Age-of-Information (AoI), formally defined as the time elapsed since the destination received the recent most update from the source. We model our scheduling problem as a variant of the multi-arm bandit problem with communication channels as arms. We characterize a lower bound on the AoI regret achievable by any policy and characterize the performance of UCB, Thompson Sampling, and their variants. Our analytical results show that UCB and Thompson sampling are order-optimal for AoI bandits. In addition, we propose novel policies which, unlike UCB and Thompson Sampling, use the current AoI to make scheduling decisions. Via simulations, we show the proposed AoI-aware policies outperform existing AoI-agnostic policies. 
\end{abstract}

\section{Introduction}
We consider a learning problem that focuses on the metric of Age of Information (AoI), introduced in \cite{kaul2012real}. AoI is formally defined as the time elapsed since the destination received the recent most update from the source. It follows that AoI is a measure of the freshness of the data available at the intended destination which makes it a suitable metric for time-sensitive systems like smart homes, smart cars, and other IoT based systems. Since its introduction, AoI has been used in areas like caching, scheduling, energy harvesting, and channel state information estimation. \footnote{A preliminary version of this work appeared in the proceedings of WiOpt 2020 \cite{bhandari2020AoI_Bandits}.} 

We focus on a system consisting of a single source that measures/tracks a time-varying quantity. The source updates a monitoring station by sending periodic updates using any one of $K$ available communication channels at a given time (Figure \ref{fig:system_model}). The probability of an attempted update succeeding is independent across communication channels and independent and identically distributed (i.i.d.) across time-slots for each channel. Channel statistics are unknown to the scheduler. AoI increases by one on each failed update and resets to one on each successful update. The goal is to determine which communication channel to use in each time-slot in order to minimize the cumulative AoI over a finite time-interval of $T$ consecutive time-slots. We view our work as a key first step towards studying real IoT-type systems which have multiple sensors updating a central monitoring station via multiple communication channels.

Like the standard multi-arm bandit (MAB) problem and its numerous variants, our scheduling problem experiences a trade-off between exploring the various communication channels and exploiting the most promising communication channel, as observed from past observations. Henceforth, we refer to our problem as \emph{AoI bandits}. The pseudo-regret of a policy at time $T$ as the difference between the cumulative AoI in the first $T$ time-slots under that policy and the cumulative AoI in the first $T$ time-slots by the ``genie" policy which uses the (statistically) best channel in each time-slot.  

The key difference between our problem and the classical MAB problem is that since AoI is correlated across time-slots, the scheduling decision made in a time-slot has a cascading effect on the regret accumulated in future time-slots. Variants of the classical MAB problem like queuing bandits \cite{krishnasamy2016learning, stahlbuhk2018learning} also exhibit this characteristic. This time-correlation has significant implications for both algorithm design and analysis. Specifically, the performance analysis of policies for AoI bandits requires a novel approach where we upper bound the regret accumulated in all future time-slots as a result of using a sub-optimal channel in a time-slot. In addition, since the potential regret accumulated in a time-slot is a function of the current AoI, it is crucial to incorporate the current AoI in making scheduling decisions. We refer to policies that do this as \emph{AoI-aware policies} and refer to policies that do not incorporate this information into their decision making as \emph{AoI-agnostic policies}. Popular policies like UCB and Thompson {\color{black}Sampling} are AoI-agnostic as they make decisions based only on the number of times each channel is used and the total number of successful transmissions on each channel. 

\begin{figure}[t]
	\centering
	\includegraphics[scale=0.45]{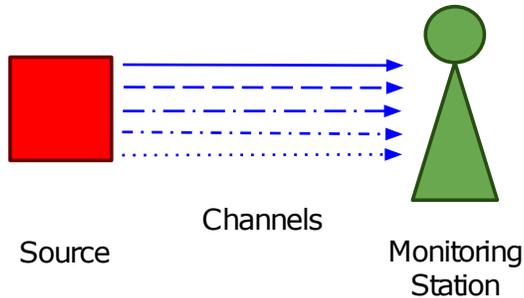}
	\caption{A system consisting of a source, a monitoring station, and five communication channels. The source tracks a time-varying quantity and sends periodic updates to the monitoring station using any one of the five channels for each update.}
	\label{fig:system_model}
\end{figure}

\subsection{Our Contributions}


\emph{Lower bound on AoI regret}: We show that the AoI regret of any $\alpha$-consistent policy is $\Omega(K\log T)$.
\color{black}

\emph{Performance of AoI-agnostic policies}: We show that the AoI regret of UCB \cite{auer2002finite} and Thompson Sampling \cite{thompson1933likelihood} is $\OO(K \log T)$, thus making them order optimal. In addition, we show that the AoI regret of Q-UCB  \cite{krishnasamy2016learning} and Q-Thompson Sampling \cite{krishnasamy2016learning} is $\OO(K \log^4 T)$.

\emph{New AoI-aware policies}: We propose variants of UCB, Thompson {\color{black}Sampling}, Q-UCB, and Q-Thompson Sampling which work in two phases. When AoI is ``low", the variants mimic the corresponding original policies and when AoI is ``high", the variants only exploit based on past observations. Via simulations, we show that the proposed variants outperform the original AoI-agnostic policies.

\subsection{Related Work}

In this section, we primarily focus on AoI based work most relevant to our setting. We refer the reader to \cite{kosta2017age} for a comprehensive survey of AoI-based works.

Scheduling to minimize AoI has been explored in a variety of settings \cite{Sombabu:2018:AAS:3241539.3267734, tripathi2017age, tripathi2019whittle, jhunjhunwala2018age, hsu2017scheduling, kadota2018optimizing}. The key difference between these works and our work is that in these works, channel statistics and/or channel state information is assumed to be known, whereas we work in the setting where channel statistics are unknown and have to be learned. In addition, some of these works focus on the infinite time-horizon and evaluate the steady-state performance, whereas we provide finite-time guarantees.

A multi-arm bandit based approach to scheduling problems to minimize queue-length is the focus of \cite{krishnasamy2016learning, stahlbuhk2018learning,cox1991queues,buyukkoc1985cmu,van1995dynamic, 
	lott2000optimality,
	ayesta2017scheduling,nino2007dynamic,mahajan2008multi,larrnaaga2016dynamic,Gittins2011}. 
	The focus in \cite{cox1991queues,buyukkoc1985cmu,van1995dynamic, 
		lott2000optimality,
		ayesta2017scheduling,nino2007dynamic,mahajan2008multi,larrnaaga2016dynamic,Gittins2011} is on the infinite horizon problem, whereas \cite{krishnasamy2016learning, stahlbuhk2018learning} focus on the finite horizon setting. The key difference between \cite{krishnasamy2016learning, stahlbuhk2018learning} and our work is that their metric is queue-length regret whereas we focus on AoI regret. We evaluate the performance of policies proposed in \cite{krishnasamy2016learning} for our metric. The policies proposed in \cite{stahlbuhk2018learning} cannot be applied in our setting due to the difference in the evolution of queue-length and AoI. 
\color{black}

A large part of the body of work focused on the AoI metric considers the setting where packets from the source(s) enter a queue and wait to be served, i.e., sent to the destination. The goal in these works is to study the effect of various queueing/service models on the resulting AoI. We refer the reader to \cite{tripathi2019age} for a comprehensive discussion of this body of work. Our setting differs from this body of work as we consider a system without a queue such that in each time-slot, the source attempts to send a fresh packet. AoI-aware scheduling in larger networks, unlike the point-to-point network studied here has also been studied \cite{shreedhar2018acp, yates2018age, yates2020age, talak2017minimizing, bedewy2019age}. 
\section{Setting}
\subsection{Our System}
We consider a system with a source and a monitoring station. The source tracks/measures a time-varying quantity and relays its measurements to the monitoring station via  $K$ communication channels as shown in Figure \ref{fig:system_model}. We use $C_i$, $1 \leq i \leq K$ to denote the $K$ channels. Time is divided into slots. In each time-slot, the source attempts to update the monitoring station by sending its current measurement via one of the $K$ communication channels. Each attempted communication via $C_i$ is successful with probability $\mu_i$ and unsuccessful otherwise, independent of all other channels and across time-slots. The values of the $\mu_i$s are unknown to the scheduler. 

\subsection{Metric: Age-of-Information Regret (AoI Regret)}

The age-of-information is a metric that measures the freshness of information available at the monitoring station. It is formally defined as follows. 


\begin{definition}[Age-of-Information (AoI)]\label{def:AoI_regret}
Let $a(t)$ denote the AoI at the monitoring station in time-slot $t$ and $u(t)$ denote the index of the time-slot in which the monitoring station received the latest update from the source before the beginning of time-slot $t$. Then, 
$a(t) = t- u(t).$
By definition, 
\begin{align*}
a(t) = 
\begin{cases}
1 & \text{if the update in slot $t-1$ succeeds} \\
a(t-1) + 1 & \text{otherwise.}
\end{cases}
\end{align*}  
\end{definition}

Let $a_{\mathcal{P}}(t)$ be the AoI in time-slot $t$ under a given {\color{black}policy} $\mathcal{P}$, and let $a^*(t)$ be the corresponding AoI under the genie policy that always uses the optimal channel, i.e., $C_{k^*}$, where
$\displaystyle k^* = \arg \max_{1\leq k \leq K} \mu_k.$ We define the AoI regret at time $T$ as the cumulative difference in expected AoI for the two policies in time-slots 1 to $T$. 
\begin{definition}[Age-of-Information Regret (AoI Regret)]
	\label{defn:regret}
	AoI regret under {\color{black}policy} $\mathcal{P}$ is denoted by $R_{\mathcal{P}}(T)$ and
\begin{align}
\label{eq:regret}
R_{\mathcal{P}}(T) = \sum_{t = 1}^{T} \E[a_{\mathcal{P}}(t) - a^*(t)].
\end{align}  
\end{definition}

For concreteness and technical convenience, we make the following assumption on the initial state of the system.

\begin{assumption}[Initial Conditions]
	\label{assumption:initialConditions}
	The system starts operating in time-slot $t = -\infty$ and the source sends an update to the monitoring station using one channel in each time-slot. Any candidate {\color{black}policy} starts making scheduling decisions at $t=1$. The {\color{black}policy} does not use information from observations in time-slots $t \leq 0$ to make decisions in time-slots $t \geq 1$. 
%
\end{assumption}

The goal is to design a scheduling policy/algorithm\footnote{We use the terms policy and algorithm interchangeably.} to minimize AoI regret (Definition \ref{defn:regret}).

\section{Main Results and Discussion}
\label{sec:mainResults}
In this section, we state and discuss our main results. A summary of our analytical results is provided in Table \ref{table:resultsSummary}. In addition to this, we propose new policies and compare the performance with known policies via simulations. 

\begin{table}[h!]
	\begin{center}
		\begin{tabular}{|l|c|} 
			\hline
		\hspace{0.37in}\textbf{Algorithm} & \textbf{Regret} \\
			\hline
				\hline
			Any $\alpha-$consistent policy & $\Omega(K \log T)$ \\
			\hline
			UCB \cite{auer2002finite} & $\OO(K \log T)$\\
			\hline
			Thompson Sampling \cite{thompson1933likelihood} & $\OO(K \log T)$ \\
			\hline
			Q-UCB \cite{krishnasamy2016learning} & $\OO(K \log^4 T)$\\
			\hline
			Q-Thompson Sampling \cite{krishnasamy2016learning} & $\OO(K \log^4 T)$ \\
			\hline
		\end{tabular}
			\caption{Summary of our analytical results}
			\label{table:resultsSummary}
	\end{center}
\end{table}

\subsection{Lower Bound on AoI Regret}
We characterize the limit on the performance of any $\alpha-$consistent policy defined as follows.
\begin{definition}{($\alpha-$consistent policies \cite{lai1985asymptotically})}
	Let $k(s)$ denote the index of the channel scheduled in time-slot $s$ and let $k^* = \arg \max_{1\leq k \leq K} \mu_k$. A scheduling policy is said to be an $\alpha-$consistent policy for $\alpha\in(0,1)$, if for any channel success probability vector $\muB$, there exists a constant $C(\muB)$ such that $$\E\left[\sum_{s=1}^{t}\Indi\{k(s)=k\}\right]\leq C(\muB)t^\alpha, \ \forall k\neq k^*.$$
\end{definition}

\begin{theorem}\label{theo:LB_AoI_regret} (Lower Bound) Given a problem instance $\muB$, let $\displaystyle \mu_{\text{min}} = \min_{i=1:K} \mu_i$, $\displaystyle \mu^* = \max_{i=1:K} \mu_i$, $\displaystyle k^* = \arg\max_{k=1:K} \mu_k$. For any $\alpha-$consistent policy $\calP$,
	\begin{align*}
	R_{\calP}(T)\geq& \frac{(K-1)D(\muB)}{\mu^*}\left((1-\alpha)\log T-\log(4KC)\right),
	\end{align*}
	where $D(\muB)=\frac{\Delta}{\text{KL}\left(\mu_{\text{min}},\frac{\mu^* +1}{2}\right)}$, $\displaystyle \Delta=\mu^{*}-\max_{k\neq k^*}\mu_k.$
\end{theorem}

We thus conclude that the AoI regret of any $\alpha-$consistent policy scales as $\Omega(K\log T)$. 

\subsection{AoI Regret of Popular AoI-agnostic Policies}

\begin{definition}[AoI-Agnostic Policies]
	A policy is AoI-agnostic if, given past scheduling decisions and the number of successful updates sent via each of the $K$ channels in the past, it does not explicitly use the AoI in a time-slot to make scheduling decisions.
\end{definition}
	
We now characterize the performance of four known AoI-agnostic policies, namely, UCB \cite{auer2002finite}, Thompson Sampling (TS) \cite{thompson1933likelihood}, Q-UCB \cite{krishnasamy2016learning}, and Q-Thompson Sampling (Q-TS) \cite{krishnasamy2016learning}. 

The UCB and Thompson {\color{black}Sampling} policies are known to perform well for MAB and Q-UCB, and Q-Thompson Sampling are variants of UCB and Thompson {\color{black}Sampling} respectively, proposed in \cite{krishnasamy2016learning} for queueing bandits. The difference between the original policies and their ``Q-" variants is that, in each time-slot, the variants force the policy to explore with a probability which decays with time, similar to the $\epsilon$-greedy {\color{black}policy} proposed in \cite{auer2002finite}. For the sake of completeness, we also provide a formal description of these policies. 

\begin{algorithm}[h]
	{
		\DontPrintSemicolon 
		\textbf{Initialise:} Set $\hat{\mu}_k=0$ to be the estimated success probability of Channel $k$, $T_k(0)=0$ $\forall$ $k\in[K]$.\;
		\While{$1\leq t \leq K$}{
			Schedule update on Channel $k(t)=t$\;
			Receive rewards $X_{k(t)}(t)\sim \text{Ber}(\mu_{k(t)})$\; $\hat{\mu}_{k(t)}=X_{k(t)}(t)$\;
			$T_{k(t)}(t)=1$\;
			$t=t+1$}
		\While{$t\geq K+1$}{
			Schedule update on Channel $k(t)=\arg \max_{k\in [K]}\hat{\mu}_k(t)+\sqrt{\frac{8\log t}{T_k(t-1)}}$ \\
			Receive reward $X_{k(t)}(t)\sim \text{Ber}(\mu_{k(t)})$\;
			$\hat{\mu}_{k(t)}=(\hat{\mu}_{k(t)}\cdot T_{k(t)}(t-1)+X_{k(t)}(t))/(T_{k(t)}(t-1)+1)$\;
			$T_{k(t)}(t)=T_{k(t)}(t-1)+1$\;
			$t=t+1$}
		\caption{{\sc Upper Confidence Bound (UCB)}}
		\label{algo:UCB}}
\end{algorithm}


{\color{black} \begin{theorem}\label{theo:UB_AoI_regret_UCB_impro}(Performance of UCB)
 	Consider any problem instance $\muB$ such that $\displaystyle k^* = \arg\max_{k=1:K} \mu_k$, $\displaystyle \mu_{\text{min}} = \min_{i=1:K} \mu_i>0$, $\displaystyle \mu^* = \max_{i=1:K} \mu_i$, and $\displaystyle \Delta = \mu^*-\max_{k\neq k^*}\mu_{k}$. Then, under Assumption 1, 
 	$$R_{\text{UCB}}(T)\leq\left\{\begin{array}{lc}
 	\frac{1-\mu^*}{\mu^*\mu_{\text{min}}} + \left(\frac{1}{\mu_{\text{min}}}-\frac{1}{\mu^{*}}\right)(K-1)\left(\frac{32 \log T}{\Delta^{2}} + 1 + \frac{\pi^{2}}{3} \right), & \text{for }T>
 	K\\
 	\left(\frac{1}{\mu_{\text{min}}}-\frac{1}{\mu^*}\right)T, & \text{for }T \leq K.
 	\end{array}\right.$$
 \end{theorem}
 We thus conclude that AoI regret of UCB scales as $\OO(K\log T)$, thus making it order optimal. The proof of Theorem \ref{theo:UB_AoI_regret_UCB_impro} upper bounds AoI regret as a function of the expected number of times sub-optimal channels are scheduled. The result then follows using a known upper bound on this quantity for UCB  \cite{auer2002finite} .}

\begin{algorithm}[h]
	\DontPrintSemicolon 
	\textbf{Initialise:} Set $\hat{\mu}_k=0$ to be the estimated success probability of Channel $k$, {$T_k(0)=0$ $\forall$ $k\in[K]$.}\;
	\While{$t\geq 1$}{
			$\alpha_{k}(t)=\hat{\mu}_k(t) T_k(t-1)+1$,\;
			$\beta_{k}(t)=(1-\hat{\mu}_k(t)) T_k(t-1)+1$,\;
			For each $k\in[K]$, pick a sample $\hat{\theta}_k(t)$ where $\hat{\theta}_k(t)\sim \text{Beta}(\alpha_{k}(t),\beta_{k}(t))$\\
	Schedule update on Channel $k(t)=\arg \max_{k\in [K]}\hat{\theta}_k(t)$\\ 
		Receive reward $X_{k(t)}(t)\sim \text{Ber}(\mu_{k(t)})$\;
		$\hat{\mu}_{k(t)}=(\hat{\mu}_{k(t)}\cdot T_{k(t)}(t-1)+X_{k(t)}(t))/(T_{k(t)}(t-1)+1)$\;
		$T_{k(t)}(t)=T_{k(t)}(t-1)+1$\;
		$t=t+1$}
	\caption{{\sc Thompson Sampling }(TS)}
	\label{algo:Ths}
\end{algorithm}


{\color{black}
	\begin{theorem}\label{theo:UB_AoI_regret_TS_impro}(Performance of TS)
	Consider any problem instance $\muB$ such that $\displaystyle k^* = \arg\max_{k=1:K} \mu_k$, $\displaystyle \mu_{\text{min}} = \min_{i=1:K} \mu_i>0$, $\displaystyle \mu^* = \max_{i=1:K} \mu_i$, and $\displaystyle \Delta = \mu^*-\max_{k\neq k^*}\mu_{k}$. Then, under Assumption 1, 
	$$R_{\text{TS}}(T)\leq\left\{\begin{array}{lc}
	\frac{1-\mu^*}{\mu^*\mu_{\text{min}}} + \left(\frac{1}{\mu_{\text{min}}}-\frac{1}{\mu^{*}}\right)O(K\log T), & \text{for }T>
	K\\
	\left(\frac{1}{\mu_{\text{min}}}-\frac{1}{\mu^*}\right)T, & \text{for }T \leq K.
	\end{array}\right.$$
\end{theorem}
We thus conclude that AoI regret of TS scales as $\OO(K\log T)$, thus making it order optimal. The proof follows on the same lines as that of Theorem \ref{theo:UB_AoI_regret_UCB_impro}.}

\begin{algorithm}[h]
	\DontPrintSemicolon 
	\textbf{Initialise:} Set $\hat{\mu}_k=0$ to be the estimated success probability of Channel $k$, $T_k(0)=0$ $\forall$ $k\in[K]$.\;
	\While{$t\geq 1$}{let $E(t)\sim\text{Ber}\left(\min\left\{ 1,3K\frac{\log^2 t}{t}\right\}\right)$\;
		
		\uIf {$E(t)=1$}{
			\textit{Explore:} Schedule update on a channel chosen uniformly at random\;
		}
		\Else{
			\textit{Exploit:} Schedule update on Channel $k(t)=\arg \max_{k\in [K]}\hat{\mu}_k(t)+\sqrt{\frac{\log^2 t}{2T_k(t-1)}}$
		}
		Receive reward $X_{k(t)}(t)\sim \text{Ber}(\mu_{k(t)})$\;
		$\hat{\mu}_{k(t)}=(\hat{\mu}_{k(t)}\cdot T_{k(t)}(t-1)+X_{k(t)}(t))/(T_{k(t)}(t-1)+1)$\;
		$T_{k(t)}(t)=T_{k(t)}(t-1)+1$\;
		$t=t+1$}
	\caption{{\sc Q-Upper Confidence Bound (Q-UCB)}}
	\label{algo:Q_UCB}
\end{algorithm}

\begin{theorem}\label{theo:UB_AoI_regret_Q_UCB}(Performance of Q-UCB)
	Consider any problem instance $\muB$ such that $\displaystyle \mu_{\text{min}} = \min_{i=1:K} \mu_i>0$, $\displaystyle \mu^* = \max_{i=1:K} \mu_i$, and  $c = \frac{-1}{\log(1-\mu^*)}$.  Under Assumption \ref{assumption:initialConditions}, there exists a constant $t_0$ such that
	\begin{align*}
	R_{\text{Q-UCB}}(T)\leq 
	\begin{cases}
	\frac{c \log T + 1 + cK \log^4 T + \OO\left(\frac{K}{T^2}\right)}{\mu_{\text{min}}}
	& \text{for }T>
	t_0\\
	\left(\frac{1}{\mu_{\text{min}}}-\frac{1}{\mu^*}\right)T, & \text{for }T\leq t_0.
	\end{cases}
	\end{align*}
\end{theorem}
\color{black}

We thus conclude that AoI regret of Q-UCB scales as $\OO(K\log^4 T)$. The proof of Theorem \ref{theo:UB_AoI_regret_Q_UCB} first characterizes the AoI regret as a function of the expected number of times a sub-optimal channel is scheduled under Q-UCB. The result then follows using results in \cite{krishnasamy2016learning}.

\begin{algorithm}[h]
	\DontPrintSemicolon 
	\textbf{Initialise:} Set $\hat{\mu}_k=0$ to be the estimated success probability of Channel $k$, $T_k(0)=0$ $\forall$ $k\in[K]$.\;
	\While{$t\geq 1$}{let $E(t)\sim\text{Ber}\left(\min\left\{ 1,3K\frac{\log^2 t}{t}\right\}\right)$\;
		
		\uIf {$E(t)=1$}{
			\textit{Explore:} Schedule update on a channel chosen uniformly at random
		}
		\Else{
			\textit{Exploit:}\;
			
				$\alpha_{k}(t)=\hat{\mu}_k(t) T_k(t-1)+1$,\;
				$\beta_{k}(t)=(1-\hat{\mu}_k(t)) T_k(t-1)+1$,\;
				For each $k\in[K]$, pick a sample $\hat{\theta}_k(t)$ where ,$\hat{\theta}_k(t)\sim \text{Beta}(\alpha_{k}(t),\beta_{k}(t))$\\
			Schedule update on Channel $k(t)=\arg \max_{k\in [K]}\hat{\theta}_k(t)$}
		Receive reward $X_{k(t)}(t)\sim \text{Ber}(\mu_{k(t)})$\;
		$\hat{\mu}_{k(t)}=(\hat{\mu}_{k(t)}\cdot T_{k(t)}(t-1)+X_{k(t)}(t))/(T_{k(t)}(t-1)+1)$\;
		$T_{k(t)}(t)=T_{k(t)}(t-1)+1$\;
		$t=t+1$}
	\caption{{\sc Q-Thompson Sampling }(Q-TS)}
	\label{algo:Q_Ths}
\end{algorithm}

\begin{theorem}\label{theo:UB_AoI_regret_Q_THS}(Performance of Q-TS)
	Consider any problem instance $\muB$ such that $\displaystyle \mu_{\text{min}} = \min_{i=1:K} \mu_i>0$, $\displaystyle \mu^* = \max_{i=1:K} \mu_i$, and  $c = \frac{-1}{\log(1-\mu^*)}$.  There exists a constant $t_0$ such that
	\begin{align*}
	R_{\text{Q-TS}}(T)\leq 
	\begin{cases}
	\frac{c \log T + 1 + cK \log^4 T + \OO\left(\frac{K}{T^2}\right)}{\mu_{\text{min}}}
	& \text{for }T>
	t_0\\
	\left(\frac{1}{\mu_{\text{min}}}-\frac{1}{\mu^*}\right)T, & \text{for }T\leq t_0.
	\end{cases}
	\end{align*}
\end{theorem}
\color{black}


We conclude that AoI regret of Q-Thompson Sampling scales as $\OO(K\log^4 T)$. The proof of Theorem \ref{theo:UB_AoI_regret_Q_THS} follows on the same lines as that of Theorem \ref{theo:UB_AoI_regret_Q_UCB}.
\color{black}

\subsection{Our AoI-aware Policies}

In this section, we propose AoI-aware variants of the policies discussed in the previous section. In the classical MAB with Bernoulli rewards, the contribution of a time-slot to the overall regret is upper bounded by one. Unlike the MAB, for AoI bandits, the difference between AoIs under a candidate policy and the {\color{black}genie} policy in a time-slot is unbounded. This motivates the need to take the current AoI value into account when making scheduling decisions. Intuitively, it makes sense to explore when AoI is low and exploit when AoI is high since the cost of making a mistake is much higher when AoI is high. We use this intuition to design AoI-aware policies. The key idea behind these policies is that they mimic the original policies when AoI is below a threshold and exploit when AoI is equal to or above a threshold, for an appropriately chosen threshold. 

The first two policies (Algorithms \ref{algo:ths_Thr} and \ref{algo:UCB_Thr}) are variants of Thompson Sampling and UCB  respectively. These policies maintain an estimate of the success probability of the best arm, denoted by $\hat{\mu}^*$. When AoI is not more than $\frac{1}{\hat{\mu}^*}$, the two policies mimic UCB and Thompson Sampling respectively, and exploit the ``best" arm (based on past observations) otherwise. Due to space constraints, Algorithm \ref{algo:UCB_Thr} is formally defined in the appendix. The third and fourth policies are variants of Q-UCB and Q-Thompson Sampling. When AoI is one, the two policies mimic Q-UCB and Q-Thompson Sampling respectively and exploit the ``best" arm (based on past observations) otherwise. These {\color{black}policies} are formally defined in the appendix (Algorithms \ref{algo:Q_UCB_Thr} and \ref{algo:Q_Ths_Thr}). 

In the next section, we compare the performance of all eight policies via simulations.


\begin{algorithm}[h]
	
		\DontPrintSemicolon 
		\textbf{Initialise:} Set $\hat{\mu}_k=0$ to be the estimated success probability of Channel $k$, $T_k(0)=0$ $\forall$ $k\in[K]$.\;
		\While{$t\geq 1$}{
			$\alpha_{k}(t)=\hat{\mu}_k(t) T_k(t-1)+1$,\;
			$\beta_{k}(t)=(1-\hat{\mu}_k(t)) T_k(t-1)+1$,\;
			Let $\text{limit(t)}=\min\limits_{k\in [K]}\frac{\alpha_{k}(t)+\beta_{k}(t)}{\alpha_{k}(t)}$\;
			\uIf {$a(t-1)>\text{limit(t)}$}{
				\textit{Exploit:} Select channel with highest estimated success probability
			}
			\Else{
				\textit{Explore:}\;
				For each $k\in[K]$, pick a sample $\hat{\theta}_k(t)$, where $\hat{\theta}_k(t)\sim \text{Beta}(\alpha_{k}(t),\beta_{k}(t))$ \\
				Schedule update on Channel $k(t)=\arg \max_{k\in [K]}\hat{\theta}_k(t)$}
			Receive reward $X_{k(t)}(t)\sim \text{Ber}(\mu_{k(t)})$\;
			$\hat{\mu}_{k(t)}=(\hat{\mu}_{k(t)}\cdot T_{k(t)}(t-1)+X_{k(t)}(t))/(T_{k(t)}(t-1)+1)$\;
			$T_{k(t)}(t)=T_{k(t)}(t-1)+1$\;
			$t=t+1$}
		\caption{{\sc AoI-Aware Thompson Sampling }(AA-TS)}
		\label{algo:ths_Thr}
\end{algorithm}
\section{Simulations}
\begin{figure}[h]
	\subfloat[AoI regret as a function of time for Setting 1.a]{
		\begin{minipage}[c][1\width]{
				0.5\textwidth}
			\centering
			\includegraphics[width=1\textwidth]{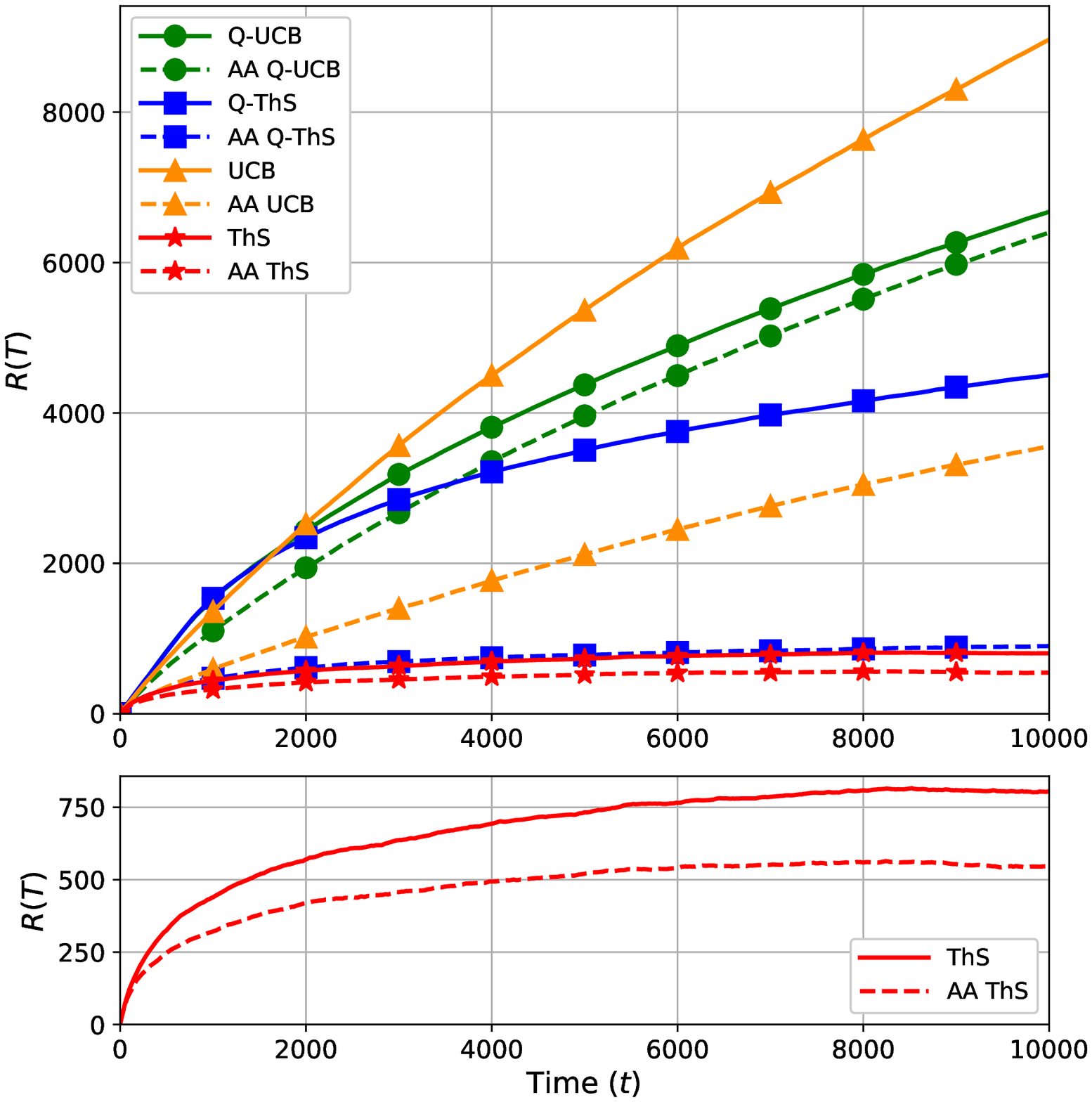}
	\end{minipage}}
	\hfill 	
	\subfloat[AoI regret as a function of time for Setting 1.b]{
		\begin{minipage}[c][1\width]{
				0.5\textwidth}
			\centering
			\includegraphics[width=1\textwidth]{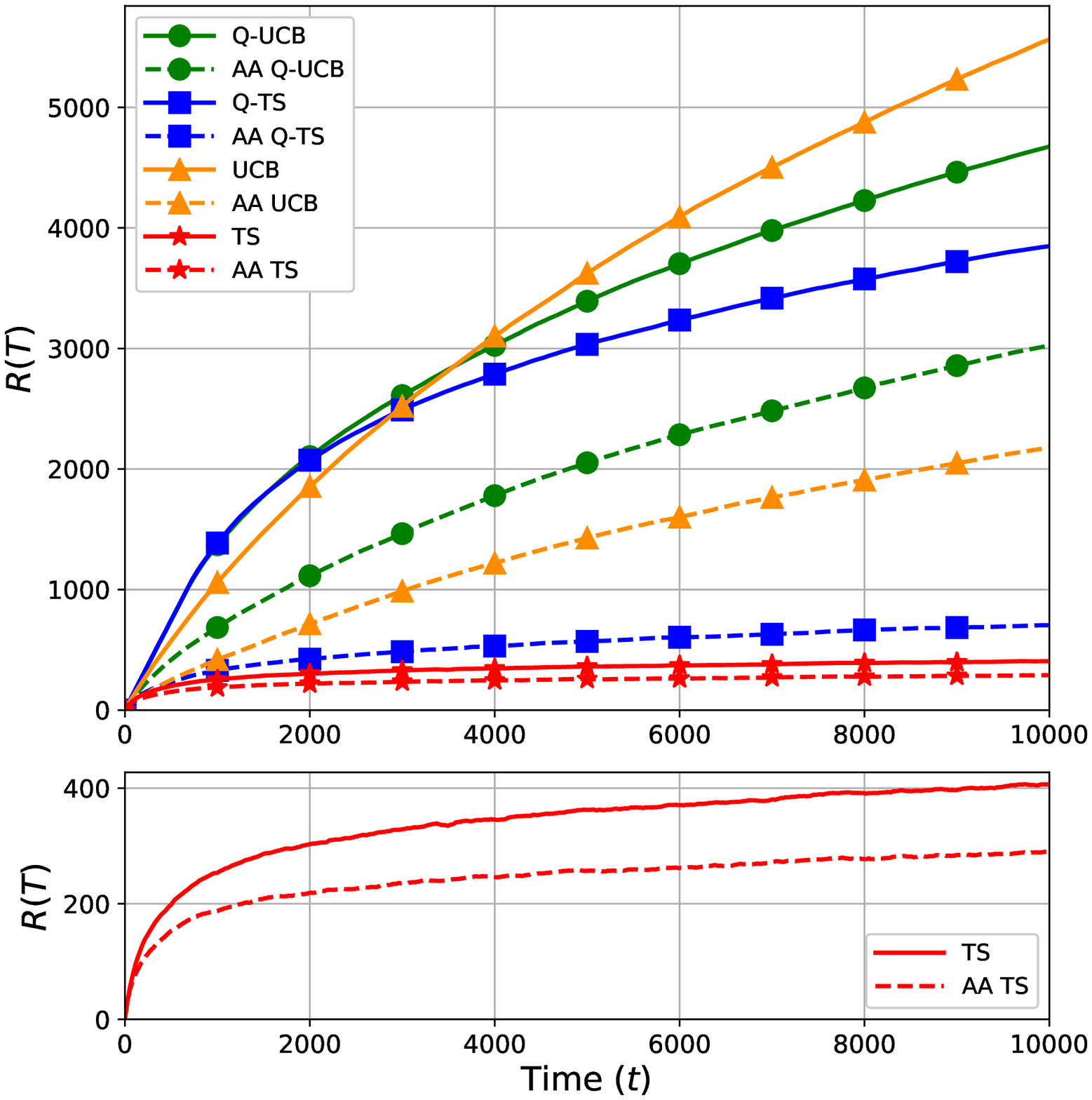}
	\end{minipage}}
	%
	\newline 
	\subfloat[AoI regret as a function of time for Setting 1.c]{
		\begin{minipage}[c][1\width]{
				0.5\textwidth}
			\centering
			\includegraphics[width=1\textwidth]{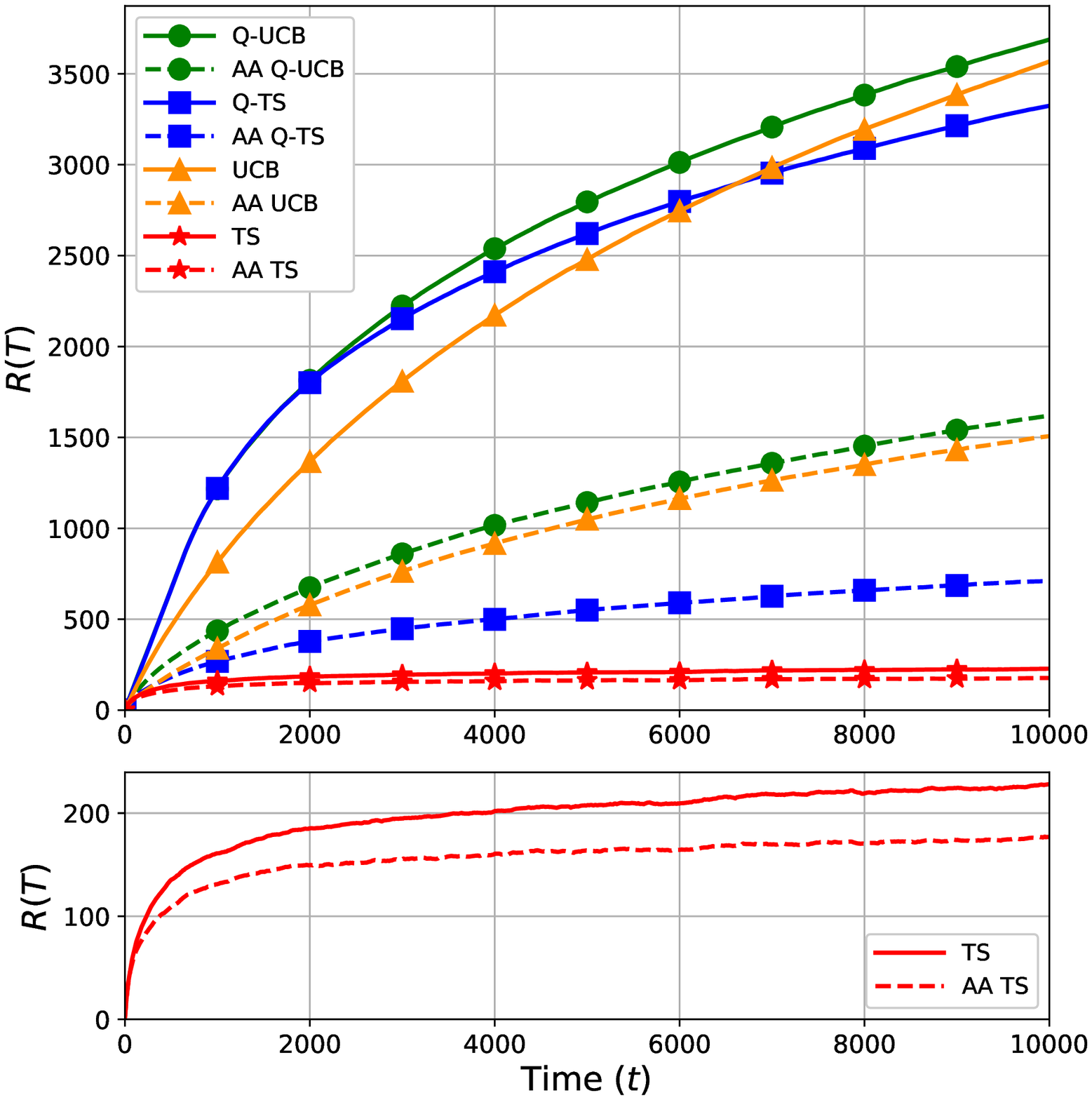}
	\end{minipage}}
	\hfill 	
	\subfloat[AoI regret as a function of time for Setting 1.d]{
		\begin{minipage}[c][1\width]{
				0.5\textwidth}
			\centering
			\includegraphics[width=1\textwidth]{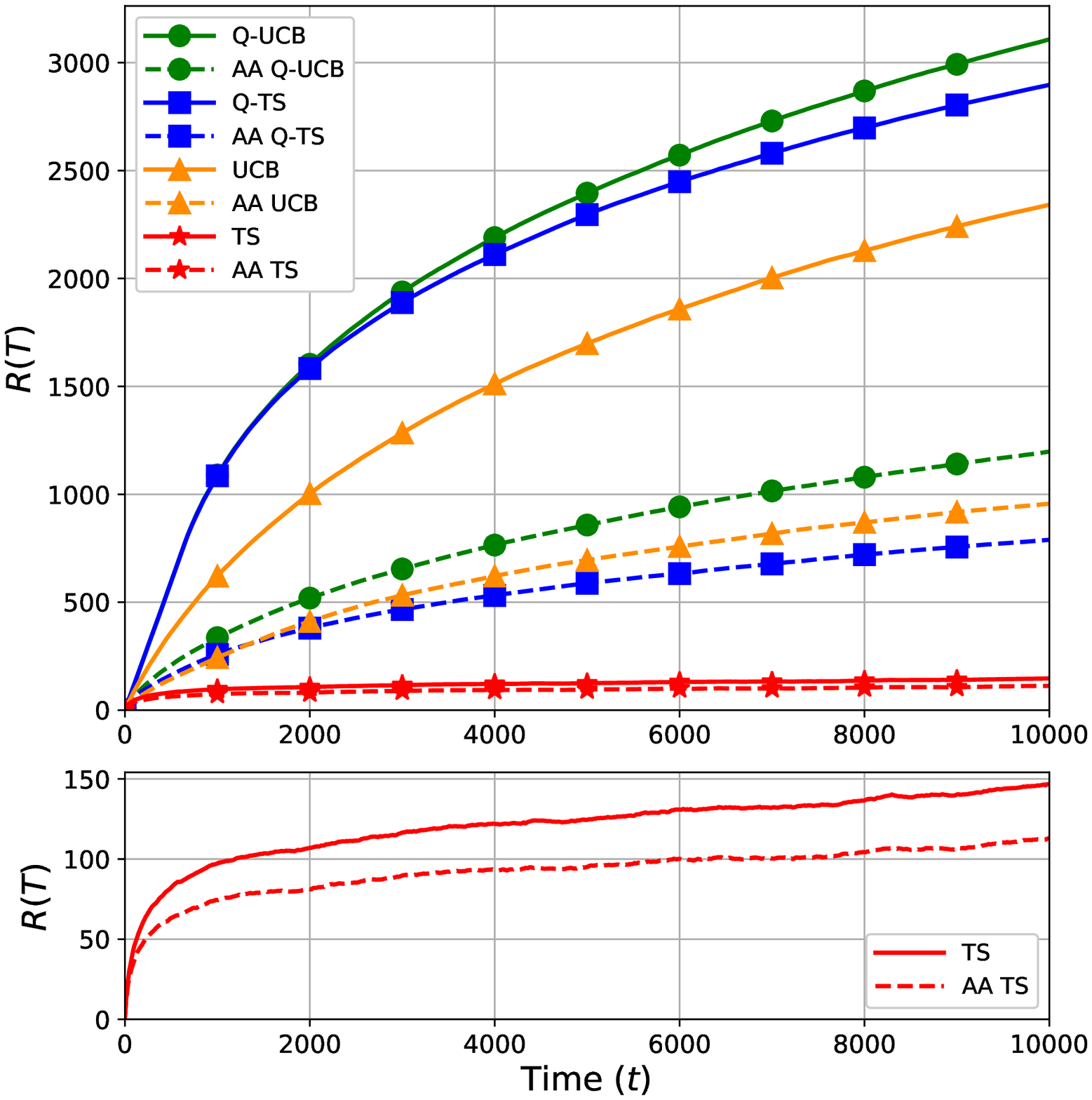}
	\end{minipage}}
	\caption{}
	\label{fig:settiing1}
\end{figure}

\begin{figure}[ht]
	\subfloat[AoI regret as a function of time for Setting 1.e]{
		\begin{minipage}[c][1\width]{
				0.5\textwidth}
			\centering
			\includegraphics[width=1\textwidth]{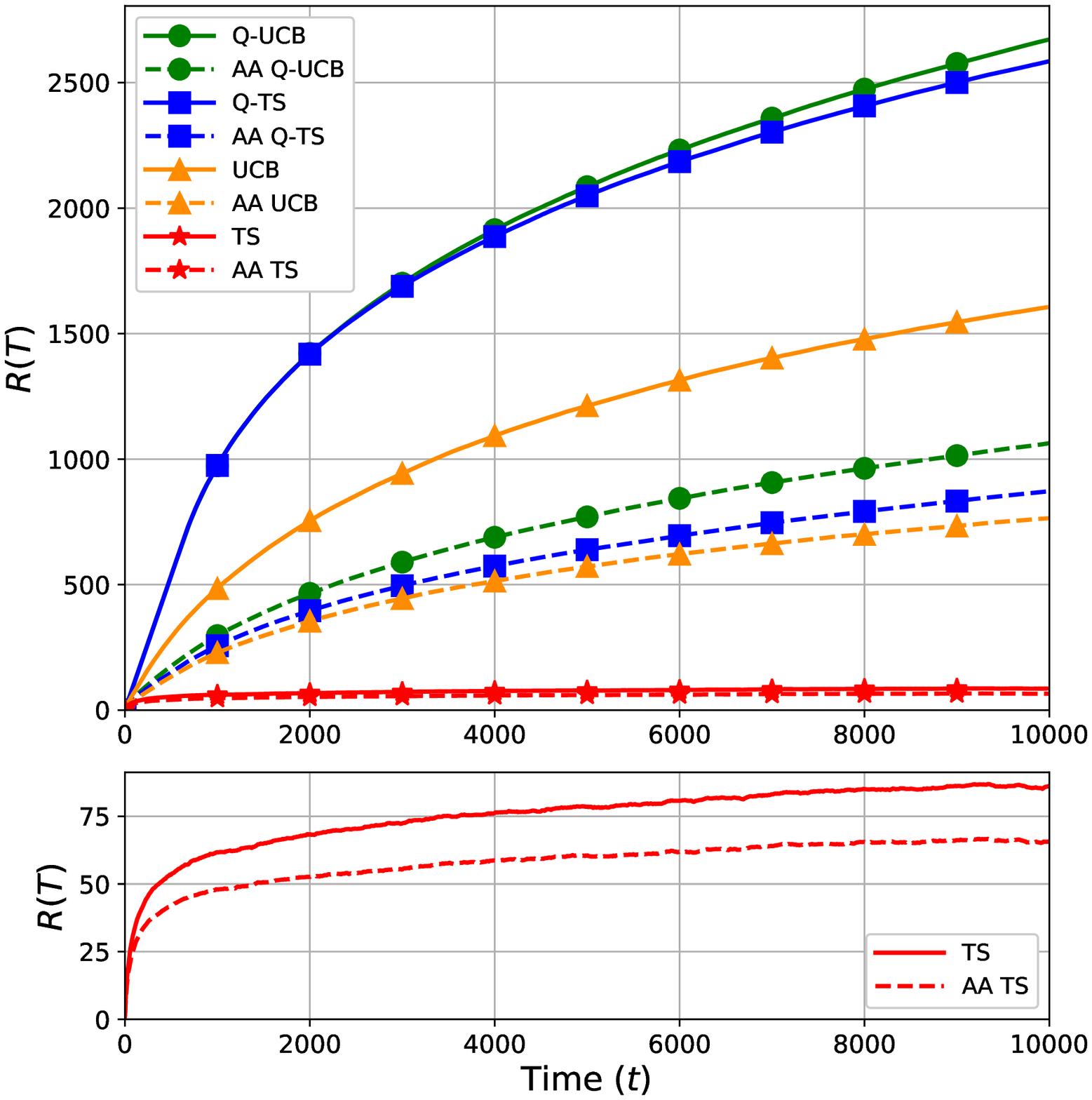}
	\end{minipage}}
	\hfill 	
	\subfloat[AoI regret as a function of time for Setting 2.a]{
		\begin{minipage}[c][1\width]{
				0.5\textwidth}
			\centering
			\includegraphics[width=1\textwidth]{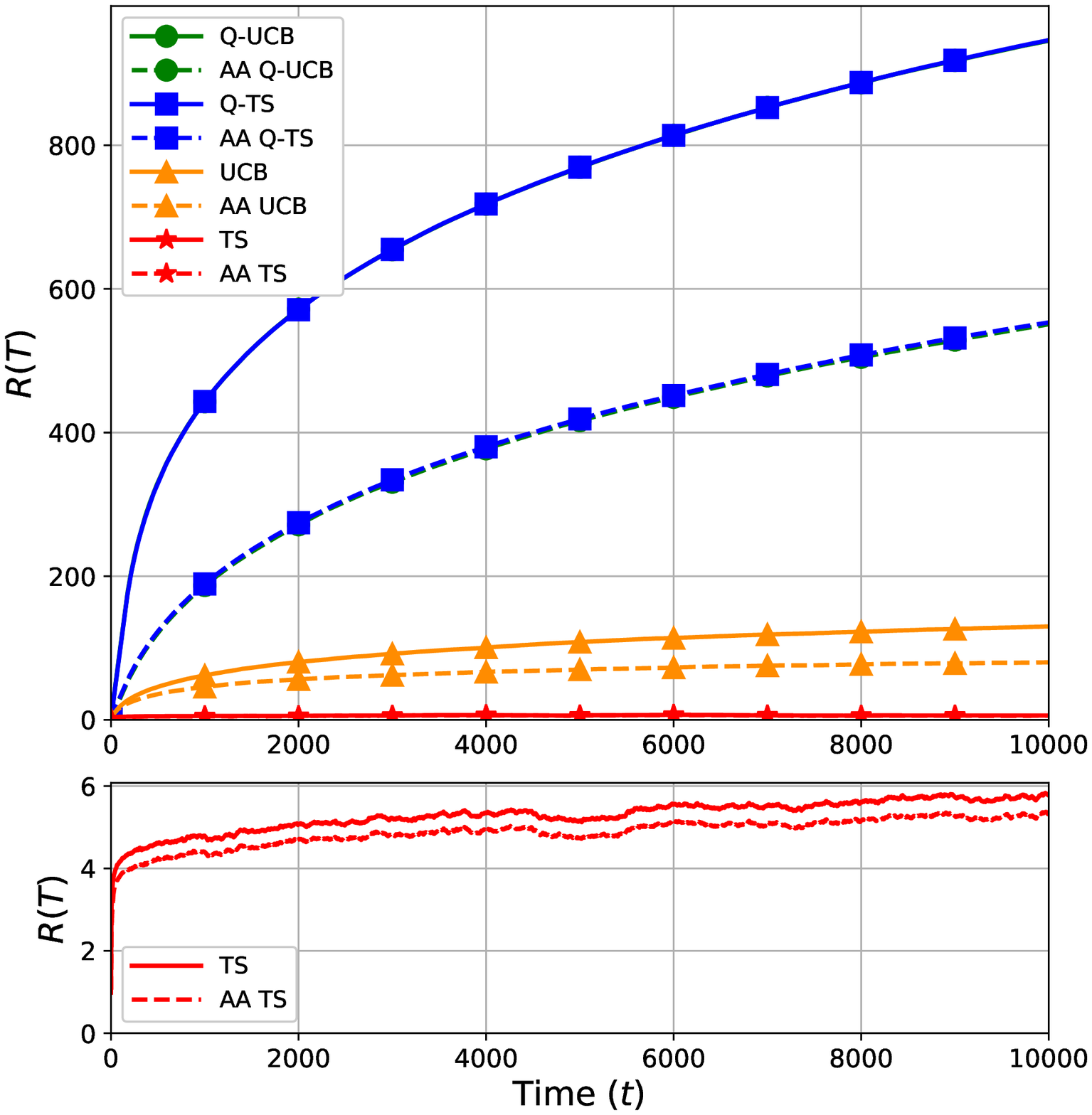}
	\end{minipage}}
	%
	\newline 
	\subfloat[AoI regret as a function of time for Setting 2.b]{
		\begin{minipage}[c][1\width]{
				0.5\textwidth}
			\centering
			\includegraphics[width=1\textwidth]{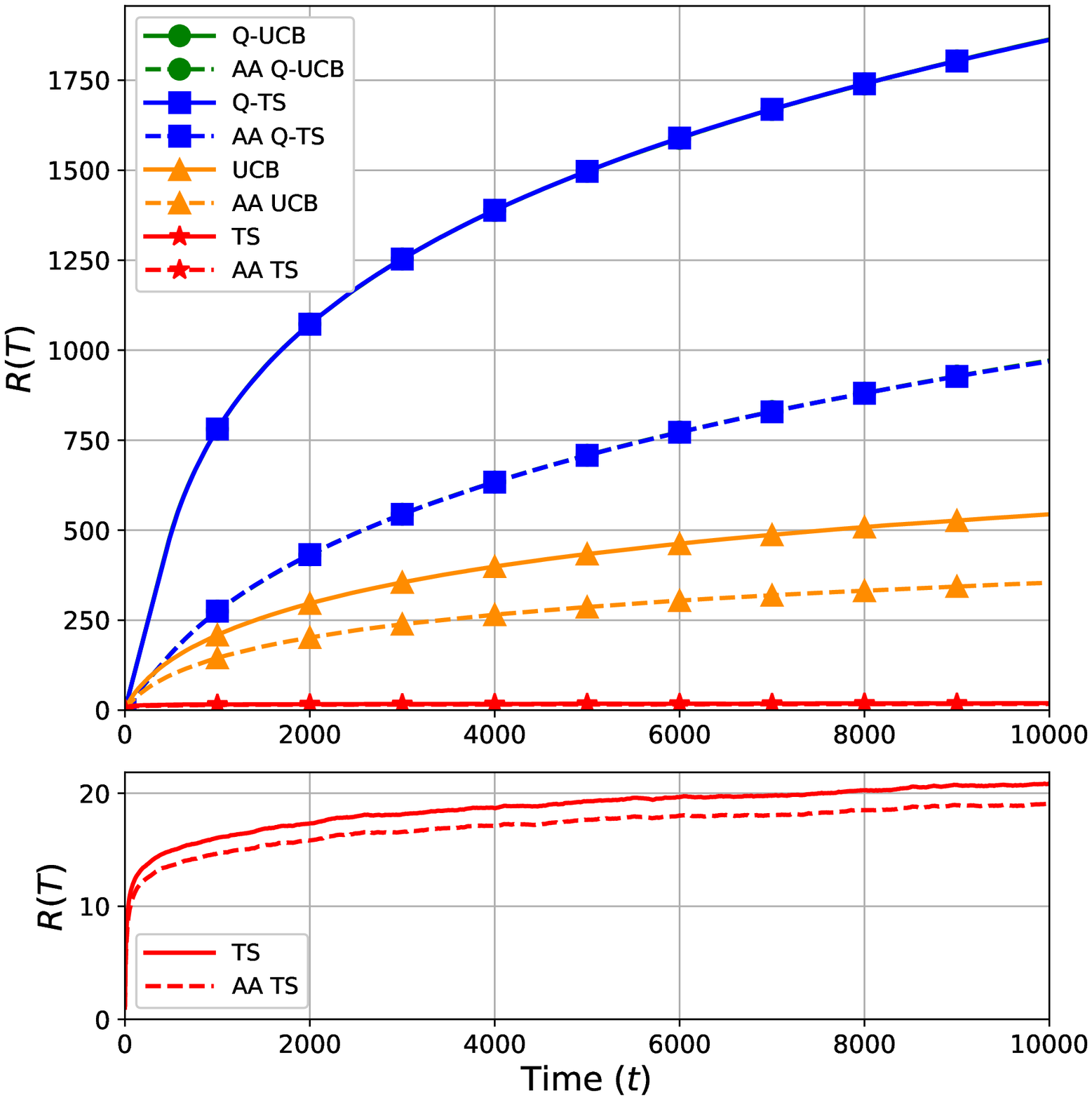}
	\end{minipage}}
	\hfill 	
	\subfloat[AoI regret as a function of time for Setting 2.c]{
		\begin{minipage}[c][1\width]{
				0.5\textwidth}
			\centering
			\includegraphics[width=1\textwidth]{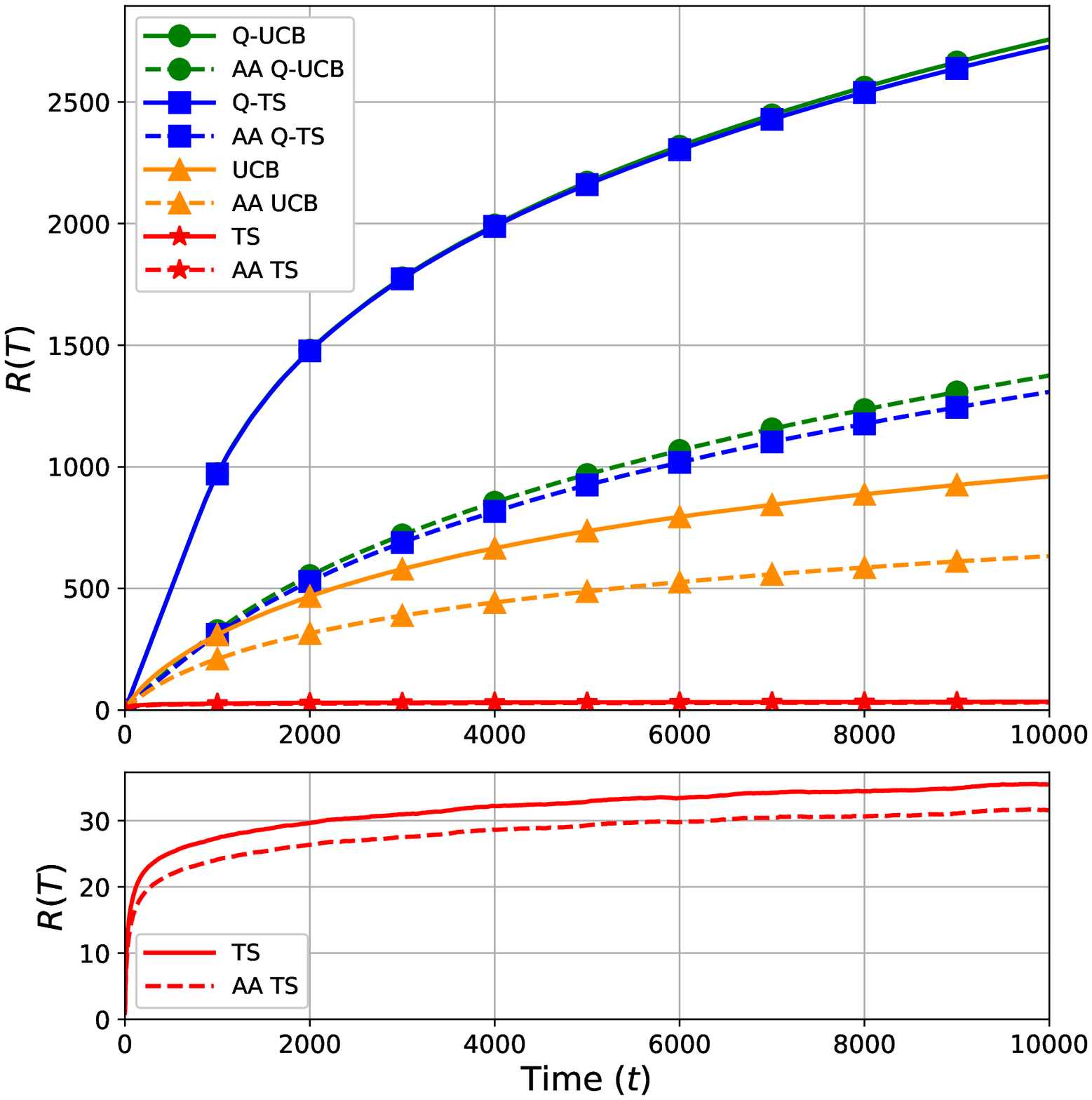}
	\end{minipage}}
	\caption{}
	\label{fig:settiing1and2}
\end{figure}

\begin{figure}[ht]
	\subfloat[AoI regret as a function of time for Setting 2.d]{
		\begin{minipage}[c][1\width]{
				0.5\textwidth}
			\centering
			\includegraphics[width=1\textwidth]{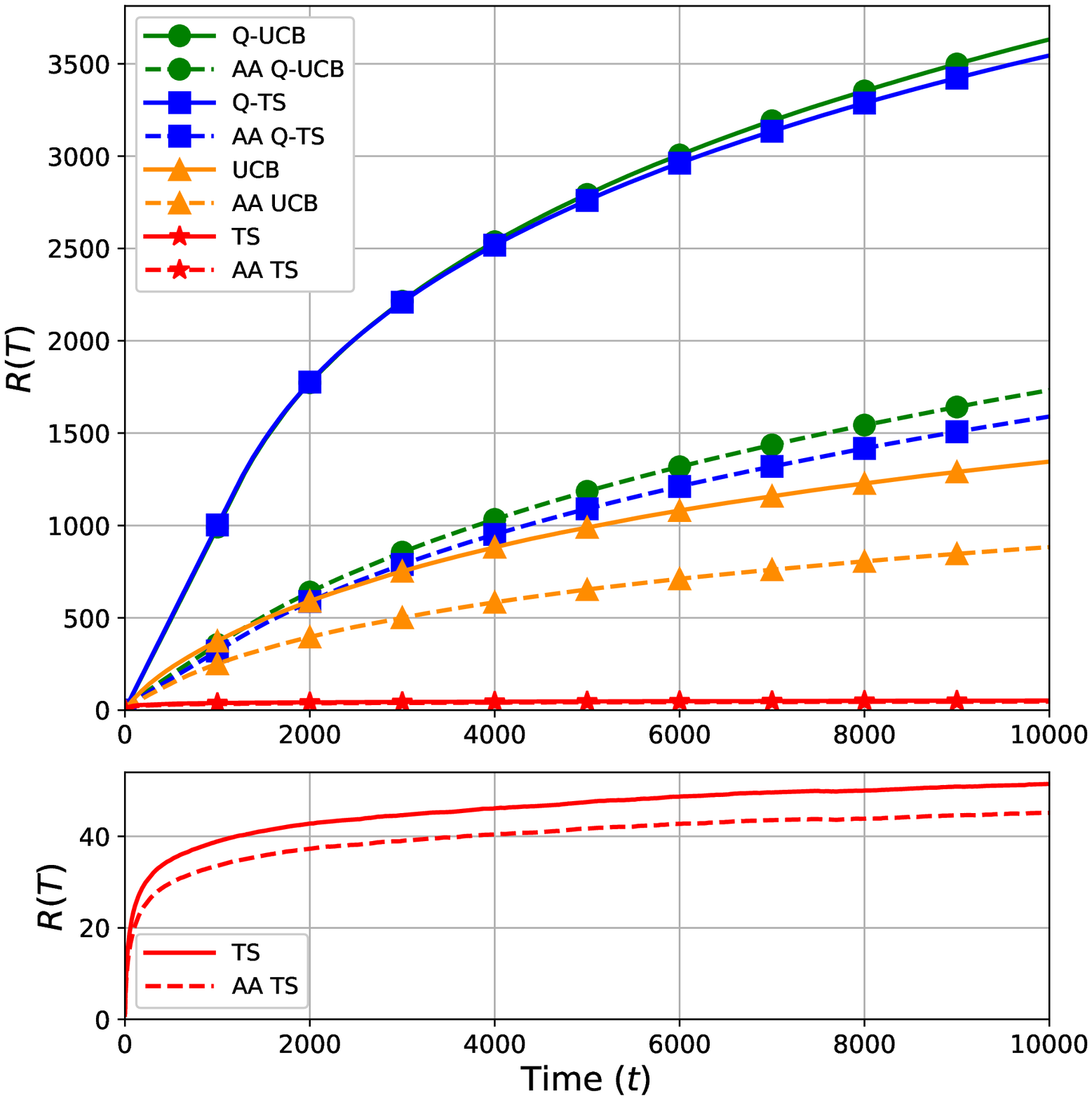}
	\end{minipage}}
	\hfill 	
	\subfloat[AoI regret as a function of time for Setting 2.e]{
		\begin{minipage}[c][1\width]{
				0.5\textwidth}
			\centering
			\includegraphics[width=1\textwidth]{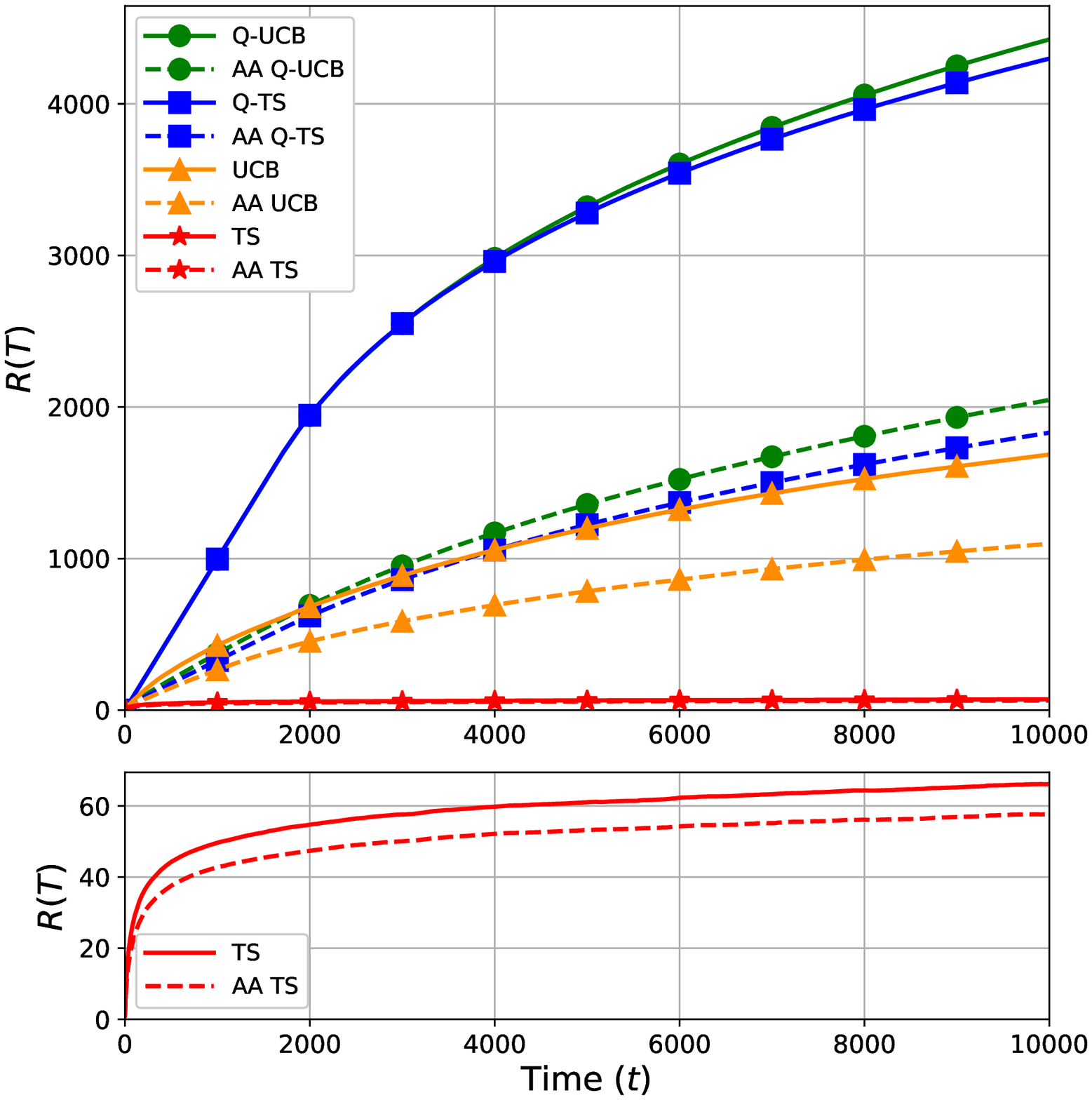}
	\end{minipage}}
	%
	\newline 
	\subfloat[AoI regret at $T=10000$ for Settings 1.a -- 1.e]{
		\begin{minipage}[c][1\width]{
				0.5\textwidth}
			\centering
			\includegraphics[width=1\textwidth]{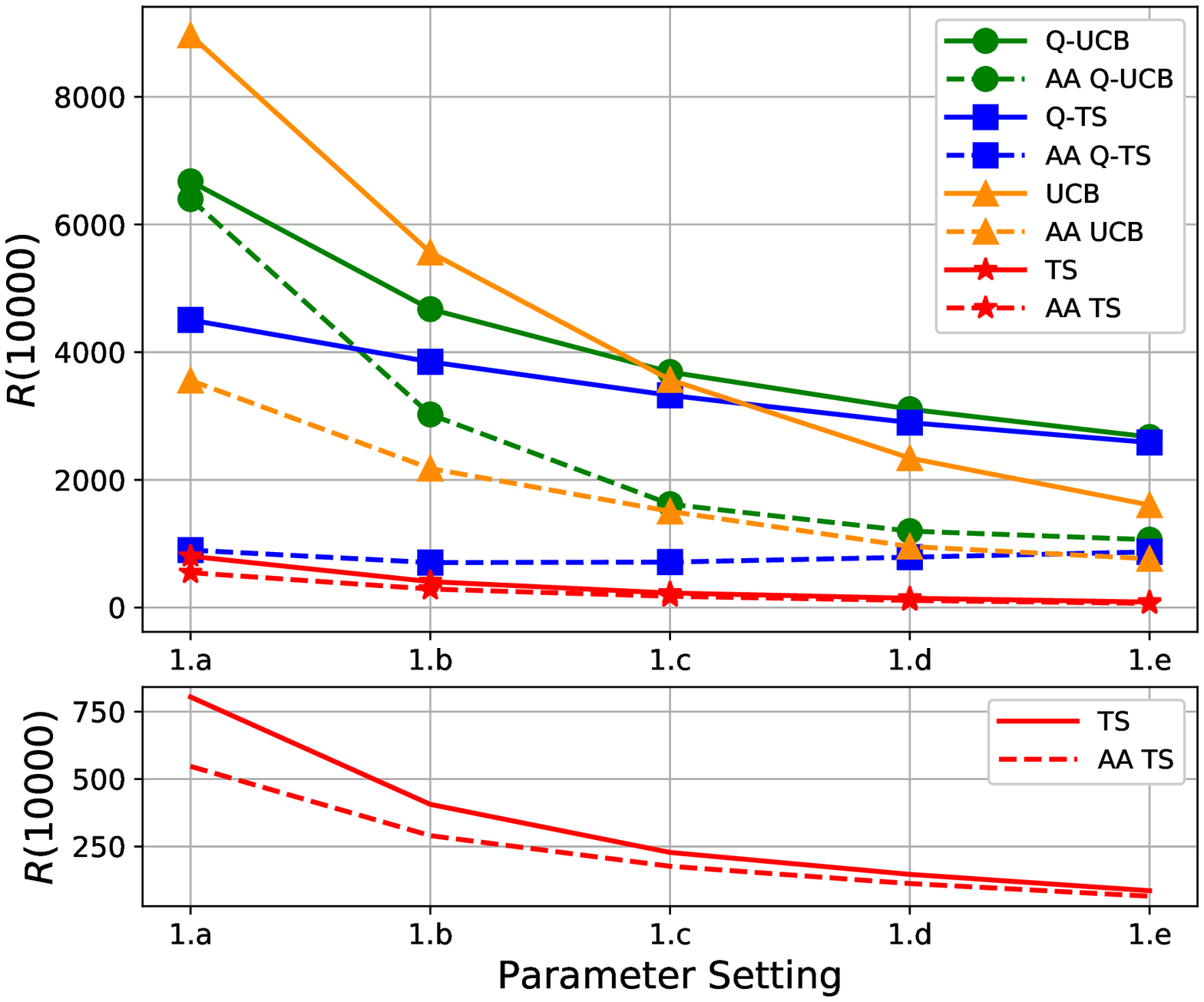}
	\end{minipage}}
	\hfill 	
	\subfloat[AoI regret at $T=10000$ for Settings 2.a -- 2.e]{
		\begin{minipage}[c][1\width]{
				0.5\textwidth}
			\centering
			\includegraphics[width=1\textwidth]{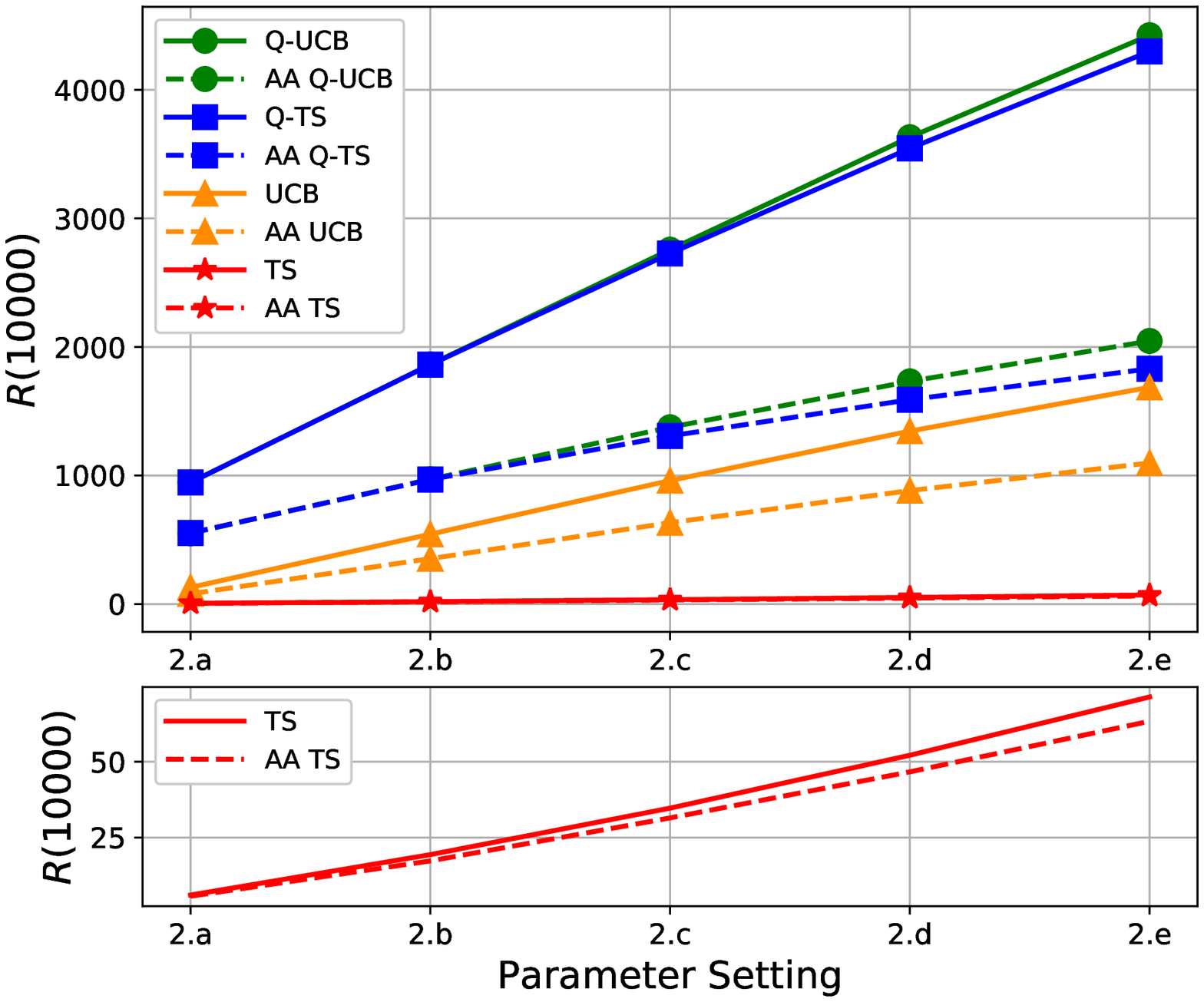}
	\end{minipage}}
	\caption{}
	\label{fig:settiing2andcombined}
\end{figure}

We present two sets of simulation results, each with five settings. In the first set, we fix the number of arms to five and vary the range of success probabilities of these five arms (Figures \ref{fig:settiing1} and \ref{fig:settiing1and2}(a)). The success probability of the five arms is equally spaced in this range, for example, if the range is $0.1$ to $0.3$, the success probabilities for the five arms are $\{0.1,0.15,0.2,0.25,0.3\}.$  In the second set of results, we consider the last five parameter settings in Table \ref{table:resultsSummary_Set2}. We fix the range of success probabilities and vary the number of arms (Figures \ref{fig:settiing1and2}(b), \ref{fig:settiing1and2}(c), \ref{fig:settiing1and2}(d), \ref{fig:settiing2andcombined}(a), and \ref{fig:settiing2andcombined}(b)). As in the first set of simulations, the success probability of the arms is equally spaced in the specified range. Each reported data-point is the average value of $1000$ independent iterations.

\begin{table}[h!]
	\begin{center}
		\begin{tabular}{|c|c|c|} 
			\hline
			\textbf{Setting} & \textbf{Range} & \textbf{Number of Arms ($K$)} \\
			\hline
			\hline
			1.a & $[0.1;0.3]$ & 5\\
			\hline
			1.b & $[0.1;0.4]$ &5\\
			\hline
			1.c & $[0.1;0.5]$ &5\\
			\hline
			1.d & $[0.1;0.6]$ &5\\
			\hline
			1.e & $[0.1;0.7]$ & 5\\
			\hline
%
			2.a & $[0.05;0.9]$ & 2\\
			\hline
			2.b & $[0.05;0.9]$ &4\\
			\hline
			2.c & $[0.05;0.9]$ &6\\
			\hline
			2.d & $[0.05;0.9]$ &8\\
			\hline
			2.e & $[0.05;0.9]$ & 10\\
			\hline
		\end{tabular}
		\caption{Simulation parameters settings. The success probability of the arms is equally spaced in the specified range.}
		\label{table:resultsSummary_Set2}
	\end{center}
\end{table}

We show the time-evolution of regret for two of the five settings in each set. In addition, we show the regret at $T=10000$ for all five settings in each set (Figures \ref{fig:settiing2andcombined}(c) and \ref{fig:settiing2andcombined}(d)). 


Consistent with expectations, AoI-agnostic {policies} are outperformed by their AoI-aware versions across all settings. The most notable observation is that AA-TS consistently outperforms all other {policies}, followed closely by TS. Also, the performance of AA-TS improves significantly relative to TS as the uniform gap between the success probabilities decreases, i.e., the optimal channel becomes harder to find. {\color{black}Notably, TS and Q-TS always outperform UCB and Q-UCB respectively, for all settings considered. Q-TS performs significantly worse than TS, but the same does not always hold true for Q-UCB and UCB respectively, which only seem to follow this trend for a sufficiently high success probability of the optimal channel.

\section{Proofs}
\label{sec:proofs}

In this section, we discuss the proofs of the results presented in Section \ref{sec:mainResults}.


\subsection{Proof of Theorem \ref{theo:LB_AoI_regret}}

To prove this theorem, we construct an alternative service process described in \cite{krishnasamy2016learning}, such that under any scheduling policy, the AoI evolution for this system has the same distribution as that for the original system. The service process is constructed as follows: let $\{U(t)\}_{t\geq 1}$ be i.i.d random variables distributed uniformly in $(0,1)$. Let the service process for Channel $k$ be given by $R_k(t)=\Indi\{U(t)\leq \mu_k\}$  for all $t$. Note that $\E[R_{k}(t)]=\mu_k$, i.e., the marginals of the service offered by each channel under this constructions is the same as that in the original system.


The proof of the claim that for any scheduling policy, the AoI evolution for this system with coupled service processes across channels has the same distribution as that for the original system follows using arguments from Section 8.1 in \cite{krishnasamy2016learning}.
We use the following result from \cite{krishnasamy2016learning} to prove Theorem \ref{theo:LB_AoI_regret}.


\begin{lemma}[Corollary 20,\cite{krishnasamy2016learning}]\label{corr_1}
	Let $T_k(t)$ be the number of time-slots in which Channel $k$ is used in the time-interval 1 to $t-1$. For a problem instance $\muB$, let $\displaystyle \mu_{\text{min}} = \min_{i=1:K} \mu_i>0$ and $\displaystyle \mu^* = \max_{i=1:K} \mu_i$. For any $\alpha-$consistent policy $\calP$, there exist constants
	$\tau$ and $C$, s.t. for any $t >\tau$, \color{black} 
 \begin{align}
	\Delta\sum_{k\neq k^*}\E\left[T_k(t+1) \right] \geq (K-1)D(\muB)\left((1-\alpha)\log t\right.\nn\left.-\log (4KC)\right),
	\end{align}
	where $D(\muB)=\frac{\Delta}{\text{KL}\left(\mu_{\text{min}},\frac{\mu^* +1}{2}\right)}$, and ${\color{black}\displaystyle \Delta=\mu^{*}-\max_{k\neq k^*}\mu_k}.$
\end{lemma}

\begin{proof}[Proof of Theorem \ref{theo:LB_AoI_regret}] Let the AoI in time-slot $t$, under an $\alpha-$consistent policy and the {\color{black}genie} policy be denoted by $a(t)$ and $a^*(t)$ respectively. Let $S(t)$ and $S^*(t)$ be indicator random variables denoting successful updates in time-slot $t$ by an $\alpha-$consistent policy and the {\color{black}genie} policy respectively. By definition, 
$
	a(t)=(1-S(t))(a(t-1)+1)+S(t), 
	a^*(t)=(1-S^*(t))(a^*(t-1)+1)+S^*(t).
$
			It follows that
	$
	a(t)-a^*(t)=(1-S(t))(a(t-1)+1)+S(t) 
	 -(1-S^*(t))(a^*(t-1)+1)-S^*(t).
$
	In the coupled system, $a^*(t)\leq a(t)$, for all $t$. Therefore,
$
	a(t)-a^*(t)\geq 
	(S^*(t)-S(t))(a^*(t-1)).
$
	Taking expectations, it follows that
	$\E\left[a(t)-a^*(t)\right]\geq \E\left[S^*(t)-S(t)\right]\cdot \E\left[a^*(t-1)\right],$
	as $a^*(t-1)$ is independent of $S^*(t)$ and $S(t)$.
	Since the {\color{black}genie} policy always uses the best channel, $a^*(t)$ is a geometric random variable with parameter $\mu^*$. It follows that $\E\left[a^*(t)\right]=(\mu^*)^{-1}$, and therefore,
	\begin{align}
	\label{eq:R_lowerBound}
R_{\calP}(T)&\geq \frac{1}{\mu^*}\sum_{t=1}^{T}\E\left[S^*(t)-S(t)\right].
	\end{align}
	
	Let $Y_k(t)$ be an indicator random variable denoting if an update sent on Channel $k$ in time-slot $t$ will be successful. Let $Y^*(t)$ be an indicator random variable denoting if an update sent on the optimal channel in time-slot $t$ will be successful. Let $k(t)$ by the index of the channel used by the $\alpha-$consistent policy in time-slot $t$.
It follows that $S^*(t)=Y^*(t) \text{ and } S(t)=\sum_{k=1}^{K}\Indi\{k(t)=k\}Y_k(t).$
Therefore,
	\begin{align}
	\label{eq:diff_S}
	\E\left[S^*(t)-S(t)\right] =& \E\left[\sum_{k\neq k^*}\Indi\{k(t)=k\}(Y^*(t) - Y_k(t))\right] \nonumber \\
	= &\sum_{k\neq k^*}\left(\PP\left(\Indi\{k(t)=k\}=1\right) \PP\left(\mu_k<U(t)\leq\mu^*\right)\right)\nn\\
	=&\sum_{k\neq k^*}(\mu^* - \mu_k)\PP\left(\Indi\{k(t)=k \}=1\right) 
	\geq \Delta\sum_{k\neq k^*}\PP\left(\Indi\{k(t)=k \}=1\right).
	\end{align}
	From \eqref{eq:R_lowerBound} and \eqref{eq:diff_S},
	\begin{align}
	R_{\mathcal{P}}(T) 
	&\geq \frac{\Delta}{\mu^*}\sum_{t=1}^{T}\sum_{k\neq k^*}\PP\left(\Indi\{k(t)=k \}=1\right)=\frac{\Delta}{\mu^*}\sum_{k\neq k^*}\E\left[T_k(T+1)\right].
	\label{eq:lower_bound}
	\end{align}
	By Lemma \ref{corr_1} and \eqref{eq:lower_bound},
	\begin{align}
R_{\calP}(T)\geq& \frac{(K-1)D(\muB)}{\mu^*}\left((1-\alpha)\log T-\log(4KC)\right).\nn
	\end{align}
\end{proof}
\color{black}

\subsection{Proofs of Theorems \ref{theo:UB_AoI_regret_UCB_impro} and \ref{theo:UB_AoI_regret_TS_impro}}

In this section, we discuss the proofs of Theorems~\ref{theo:UB_AoI_regret_UCB_impro} and~\ref{theo:UB_AoI_regret_TS_impro}. We first provide an outline of these proofs. 

\subsubsection{Proof Outline} The proof uses the following arguments.
\begin{itemize}
\item[--] We first upper bound the expected cumulative AoI for any schedule by the expected cumulative AoI of an alternative schedule (Schedule A) in which all uses of sub-optimal channels in the original schedule are replaced by using the worst channel (channel with parameter $\mu_{\text{min}}$). (Lemma \ref{lemma:systemA})
\item[--] We further upper bound the expected cumulative AoI of Schedule A with a second alternative schedule (Schedule B) where all uses of the worst channel are clustered together starting from $T=1$, followed by all uses of the optimal channel. (Lemma \ref{lemma:systemA})
\item[--] We upper bound the expected cumulative AoI of Schedule B as a function of the length of the schedule and the expected number of uses of the worst channel. (Lemma \ref{lemma:expectedAgeBound2})
\item[--] We then substitute known bounds on the expected number of uses of sub-optimal channels under UCB and Thompson Sampling to get the desired results. (Lemma \ref{lemma:UCBbound})
\end{itemize}
\subsubsection{Proof Details}
\begin{lemma}
	\label{lemma:twoCases}
	Let $k \geq 0$, $k_1 \geq 0$ and $k_2 \geq 1$ be integers and $T = k+k_1+k_2+1$.
	Consider two sequences denoting the channels scheduled in time-slots one to $T$, denoted by $\textbf{K}_I(T)$ and $\textbf{K}_{II}(T)$. 
	The channels scheduled in the first $k$ time-slots are identical in $\textbf{K}_I(T)$ and $\textbf{K}_{II}(T)$. 
	
	In Case I, the optimal channel is scheduled in time-slot $k+1$, the worst channel (channel with parameter $\mu_{\text{min}}$) is scheduled in time-slots $k+2$ to $k+k_2+1$, and the optimal channel scheduled in time-slots $k+k_2+2$ to $k+k_1+k_2+1$.
	
	In Case II, the worst channel is scheduled in time-slots $k+1$ to $k+k_2$, and the optimal channel scheduled in time-slots $k+k_2+1$ to $k+k_1+k_2+1$.
	
	Let $\E_{I}[a(t)\vert\textbf{K}_I(T)]$ and $\E_{II}[a(t)\vert\textbf{K}_{II}(T)]$ denote the expected AoI in time-slot $t$ in the two cases.  Then, 
	\begin{align}
	\sum_{m=0}^{T-1}\E_{I}[a(T-m)\vert\textbf{K}_I(T)]&\leq \sum_{m=0}^{T-1}\E_{II}[a(T-m)\vert\textbf{K}_{II}(T)]\nn.
	\end{align}
\end{lemma}
\begin{proof}
	By definition, $\PP(a(t)>\tau)=\prod_{i=0}^{\tau}(1-\mu_{k(t-i)}),$ and $$\E_{I}[a(t)\vert\textbf{K}_I(T)]=\sum_{\tau=0}^{\infty}\prod_{i=0}^{\tau}(1-\mu_{k(t-i)}).$$
	For $0 \leq m \leq k_1-1$, the expected AoI at time $t$ for the two cases are:
	\begin{align}
	\E_{I}[a(T-m)\vert \textbf{K}_{I}(T)]&=1+\sum_{i=1}^{k_1 -m}(1-\mu^*)^i+(1-\mu^*)^{k_1-m}\sum_{j=1}^{k_2}(1-\mu_{\text{min}})^{j}\nn\\
	&\hspace{0.25in}+(1-\mu^*)^{k_1-m+1}(1-\mu_{\text{min}})^{k_2}+c_m,\nn\\
	\E_{II}[a(T-m)\vert \textbf{K}_{II}(T)]&=1+\sum_{i=1}^{k_1 -m+1}(1-\mu^*)^i+(1-\mu^*)^{k_1-m+1}\sum_{j=1}^{k_2}(1-\mu_{\text{min}})^{j}+c_m,\nn
	\end{align}
	where $c_m$ is a function of the channels scheduled in time-slots $-\infty$ to $k$.
	It follows that
	\begin{align}
	&\E_{I}[a(T-m)\vert \textbf{K}_{I}(T)]-\E_{II}[a(T-m)\vert \textbf{K}_{II}(T)]\nn\\
	&\hspace{0.5in}=(1-\mu^*)^{k_1-m}\mu^*\sum_{j=1}^{k_2-1}(1-\mu_{\text{min}})^{j}+(1-\mu^*)^{k_1-m}((1-\mu_{\text{min}})^{k_2}-(1-\mu^*)).\label{eq:6}
	\end{align}
	For $m=k_1$,
	\begin{align}
	\E_{I}[a(T-m)\vert\textbf{K}_{I}(T)]-\E_{II}[a(T-m)\vert\textbf{K}_{II}(T)]
	=\mu^*\sum_{j=1}^{k_2-1}(1-\mu_{\text{min}})^{j}+(1-\mu_{\text{min}})^{k_2}-(1-\mu^*).\label{eq:7}
	\end{align}
	Combining \eqref{eq:6} and \eqref{eq:7}, for $0 \leq m \leq k_1$,
	\begin{align}
	&\E_{I}[a(T-m)\vert\textbf{K}_{II}(T)]-\E_{II}[a(T-m)\vert\textbf{K}_{II}(T)] \nn \\
	&=(1-\mu^*)^{k_1-m}\mu^*\sum_{j=1}^{k_2-1}(1-\mu_{\text{min}})^{j}+(1-\mu^*)^{k_1-m}((1-\mu_{\text{min}})^{k_2}-(1-\mu^*)).\label{eq:8}
	\end{align}
	For $k_1+1 \leq m \leq k_1+k_2-1$,
	\begin{align}
	E_{I}[a(T-m)\vert\textbf{K}_{I}(T)]&=1+\sum_{j=1}^{k_2-m+k_1}(1-\mu_{\text{min}})^j+(1-\mu_{\text{min}})^{k_2-m+k_1}(1-\mu^*)\nn\\
	&\hspace{0.25in}+(1-\mu_{\text{min}})^{k_2-m+k_1}(1-\mu^*)\sum_{j=k_2-m+k_1+2}^{\infty}\;\prod_{i=k_2-m+k_1+2}^{j}(1-\mu_{k(T-m-i)}),\nn\\
	\E_{II}[a(T-m)\vert\textbf{K}_{II}(T)]&=1+\sum_{j=1}^{k_2-m+k_1+1}(1-\mu_{\text{min}})^{j}+(1-\mu_{\text{min}})^{k_2-m+k_1+1}\times\nn\\
	&\hspace{0.4in}\sum_{j=k_2-m+k_1+2}^{\infty}\;\prod_{i=k_2-m+k_1+2}^{j}(1-\mu_{k(T-m-i)}).\nn
	\end{align}
	Therefore,
	\begin{align}
	&\E_{I}[a(T-m)\vert\textbf{K}_{I}(T)]-\E_{II}[a(T-m)\vert\textbf{K}_{II}(T)]\nn\\
	&\hspace{0.15in}=(1-\mu_{\text{min}})^{k_2-m+k_1}(\mu_{\text{min}}-\mu^*)\left(1+\sum_{j=k_2-m+k_1+2}^{\infty}\;\prod_{i=k_2-m+k_1+2}^{j}(1-\mu_{k(T-m-i)})\right).\label{eq:9}
	\end{align}	
	For $m=k_1+k_2$,
	\begin{align}
	\E_{I}[a(T-m)\vert\textbf{K}_{I}(T)]&=1+(1-\mu^*)+(1-\mu^*)\sum_{j=k_2-m+k_1+2}^{\infty}\;\prod_{i=k_2-m+k_1+2}^{j}(1-\mu_{k(T-m-i)}),\nn\\
	\E_{II}[a(T-m)\vert\textbf{K}_{II}(T)]&=1+(1-\mu_{\text{min}})+(1-\mu_{\text{min}})\sum_{j=k_2-m+k_1+2}^{\infty}\;\prod_{i=k_2-m+k_1+2}^{j}(1-\mu_{k(T-m-i)}).\nn
	\end{align}
	Therefore,
	\begin{align}
	&\E_{I}[a(T-m)\vert\textbf{K}_{I}(T)]-
	\E_{II}[a(T-m)\vert\textbf{K}_{II}(T)]\nn\\
	&\hspace{0.15in}=(\mu_{\text{min}}-\mu^*)\left(1+\sum_{j=k_2-m+k_1+2}^{\infty}\;\prod_{i=k_2-m+k_1+2}^{j}(1-\mu_{k(T-m-i)})\right)\label{eq:10}
	\end{align}
	Combining \eqref{eq:9} and \eqref{eq:10}, for $k_1+1\leq m \leq k_1+k_2$,
	\begin{align}
	&\E_{I}[a(T-m)\vert\textbf{K}_{I}(T)]-\E_{II}[a(T-m)\vert\textbf{K}_{II}(T)]\nn\\
	&\hspace{0.15in}=(1-\mu_{\text{min}})^{k_2-m+k_1}(\mu_{\text{min}}-\mu^*)\left(1+\sum_{j=k_2-m+k_1+2}^{\infty}\;\prod_{i=k_2-m+k_1+2}^{j}(1-\mu_{k(T-m-i)})\right).\label{eq:11}
	\end{align}
	Note that in \eqref{eq:8},
	\begin{align}
	&\sum_{m=0}^{k_1}(1-\mu^*)^{k_1-m}\mu^*\sum_{j=1}^{k_2-1}(1-\mu_{\text{min}})^j+((1-\mu_{\text{min}})^{k_2}-(1-\mu^*))\sum_{m=0}^{k_1}(1-\mu^*)^{k_1-m}\nn\\
	&=\left[\frac{\mu^*(1-\mu_{\text{min}})(1-(1-\mu_{\text{min}})^{k_2-1})}{\mu_{\text{min}}}+((1-\mu_{\text{min}})^{k_2}-(1-\mu^*))\right]\times\frac{1-(1-\mu^*)^{k_1+1}}{\mu^*}.\label{eq:12}
	\end{align}
	Similarly in \eqref{eq:11},
	\begin{align}
	&\sum_{m=k_1+1}^{k_1+k_2}(1-\mu_{\text{min}})^{k_2-m+k_1}(\mu_{\text{min}}-\mu^*)\left(1+\sum_{j=k_2-m+k_1+2}^{\infty}\;\prod_{i=k_2-m+k_1+2}^{j}(1-\mu_{k(T-m-i)})\right)\nn\\
	&\hspace{0.15in}=(\mu_{\text{min}}-\mu^*)\left(1+\sum_{j=k_2-m+k_1+2}^{\infty}\;\prod_{i=k_2-m+k_1+2}^{j}(1-\mu_{k(T-m-i)})\right)\times\frac{1-(1-\mu_{\text{min}})^{k_2}}{\mu_{\text{min}}}.\label{eq:13}
	\end{align}
	Also, the expression in square brackets in \eqref{eq:12}, given by
	\begin{align}
	&\frac{\mu^*(1-\mu_{\text{min}})}{\mu_{\text{min}}}-\frac{\mu^*(1-\mu_{\text{min}})^{k_2}}{\mu_{\text{min}}}+\frac{\mu_{\text{min}}(1-\mu_{\text{min}})^{k_2}}{\mu_{\text{min}}}-\frac{\mu_{\text{min}}(1-\mu^*)}{\mu_{\text{min}}}\nn\\
	&\hspace{0.15in}=\frac{\mu^*-\mu_{\text{min}}}{\mu_{\text{min}}}\left(1-(1-\mu_{\text{min}})^{k_2}\right)\label{eq:14}
	\end{align}
	Combining \eqref{eq:12}, \eqref{eq:13} and \eqref{eq:14}, we have that
	\begin{align}
	&\frac{(\mu^*-\mu_{\text{min}})\left(1-(1-\mu_{\text{min}})^{k_2}\right)}{\mu_{\text{min}}}\left(\frac{1-(1-\mu^*)^{k_1+1}}{\mu^*}\right)\nn\\
	&\hspace{0.25in}+\frac{(1-(1-\mu_{\text{min}})^{k_2})(\mu_{\text{min}}-\mu^*)}{\mu_{\text{min}}}\left(1+\sum_{j=k_2-m+k_1+2}^{\infty}\;\prod_{i=k_2-m+k_1+2}^{j}(1-\mu_{k(T-m-i)})\right)\hspace{3in}\nn\\
	&=\frac{(\mu^*-\mu_{\text{min}})(1-(1-\mu_{\text{min}})^{k_2})}{\mu_{\text{min}}}\left[\frac{1-(1-\mu^*)^{k_1+1}}{\mu^*}-1\right.\nn\\
	&\hspace{2.5in}-\left.\sum_{j=k_2-m+k_1+2}^{\infty}\;\prod_{i=k_2-m+k_1+2}^{j}(1-\mu_{k(T-m-i)})\right].
	\end{align}
	It follows that
	\begin{align}
	&\sum_{m=0}^{k_1+k_2}\left(\E_{I}[a(T-m)\vert\textbf{K}_{I}(T)]-\E_{II}[a(T-m)\vert\textbf{K}_{II}(T)]\right)\nn\\&\hspace{0.25in}\leq{\underbrace{\frac{(\mu^*-\mu_{\text{min}})(1-(1-\mu_{\text{min}})^{k_2})}{\mu_{\text{min}}}}_{> 0}}{\underbrace{\left[\frac{1-(1-\mu^*)^{k_1+1}}{\mu^*}-\frac{1}{\mu^*}\right]}_{<0}}.\nn
	\end{align}
	Since the channels scheduled in $-\infty \leq t \leq k$ are identical in the two cases, it follows that
	\begin{align}
	\sum_{m=0}^{T-1}\E_{I}[a(T-m)\vert\textbf{K}(T)]\leq \sum_{m=0}^{T-1}\E_{II}[a(T-m)\vert\textbf{K}(T)],\nn
	\end{align}
	thus proving the result.
	%
\end{proof}

\begin{lemma}
	\label{lemma:systemA}
	
	Let $\textbf{K}(T)$ be a sequence of channels scheduled in time-slots one to $T$ and let $N(\textbf{K}(T))$ denote the number of time-slots in which a sub-optimal channel is used in time-slots one to $T$ under $\textbf{K}(T)$. Let $\textbf{K}_A(T)$ be an alternative sequence of channels scheduled in time-slots one to $T$ derived from $\textbf{K}(T)$ such that all uses of the sub-optimal channel in $\textbf{K}(T)$ are replaced by the worst channel, i.e., the channel with parameter $\mu_{\text{min}}$. Let $\textbf{K}_B(T)$ be another alternative sequence of channels scheduled in time-slots one to $T$ derived from $\textbf{K}_A(T)$ such that the worst channel is used in time-slots one to $N(\textbf{K}(T))$ and the optimal channel is used thereafter, i.e., in time-slots $N(\textbf{K}(T))+1$ to $T$. Then we have that,
	
	\begin{align}
	\sum_{m=0}^{T-1}\E[a(T-m)\vert\textbf{K}(T)]\leq \sum_{m=0}^{T-1}\E[a(T-m)\vert\textbf{K}_A(T)] \leq \sum_{m=0}^{T-1}\E[a(T-m)\vert\textbf{K}_B(T)].
	\end{align}

\end{lemma}
\begin{proof}
	
	Since, $\mu_{\text{min}} \leq \mu_{k(t)}\ \forall t$, it follows that
	\begin{align}
	\sum_{m=0}^{T-1}\E[a(T-m)\vert\textbf{K}(T)]\leq \sum_{m=0}^{T-1}\E[a(T-m)\vert\textbf{K}_A(T)]. \nn
	\end{align}
	
	Further, if $\textbf{K}_A(T) \neq \textbf{K}_B(T)$, there exists constants $k \geq 0$, $k_1 \geq 1$, and  $k_2 \geq 0$ such that $\textbf{K}_A(T)$ satisfies the conditions of Case I discussed in Lemma \ref{lemma:twoCases}. Further, using Lemma \ref{lemma:twoCases}, the expected cumulative AoI conditioned on $\textbf{K}_A(T)$ is upper bounded by the the expected cumulative AoI in the corresponding Case II sequence. We recursively apply the same argument on the sequence of channels scheduled in Case II till the Case II sequence is equal to $\textbf{K}_B(T)$. Therefore,
	\begin{align}
	\sum_{m=0}^{T-1}\E[a(T-m)\vert\textbf{K}_A(T)] \leq \sum_{m=0}^{T-1}\E[a(T-m)\vert\textbf{K}_B(T)], \nn
	\end{align}
	thus proving the result.
\end{proof}
%
\color{red}

\color{black}
\begin{lemma}
	\label{lemma:expectedAgeBound2}
	Let $k(t)$ denote the index of the communication channel used in time-slot $t$ and $k^*$ be the index of the optimal channel. Let $\textbf{K}(T) = \{k(1), k(2), \cdots, k(T)\}$ be the sequence of channels used in time-slots $1$ to $T$ and
	$$N(\textbf{K}(T)) = \sum_{t=1}^T \mathbbm{1}_{k(t) \neq k^*},$$ denote the number of time-slots in which a sub-optimal channel is used. 	
	Under Assumption \ref{assumption:initialConditions},
	\begin{align*}
	\sum_{t = 1}^{T} \E[a(t)]  \leq \frac{T}{\mu^*}+\frac{1-\mu^*}{\mu^*\mu_{\text{min}}}+\left(\frac{1}{\mu_{\text{min}}}-\frac{1}{\mu^*}\right)\E[N(\bfK).
	\end{align*}
\end{lemma}
\begin{proof}
	From Lemma \ref{lemma:systemA},
	\begin{align}
	\sum_{t=1}^{T}\E[a(t)]
	&\leq \E\left[\sum_{m=0}^{T-1}\E[a(T-m)\vert \textbf{K}_B(T),N(\textbf{K}_B(T))=n]\right].\label{eq:15}
	\end{align}
	Note that for $0 \leq m \leq T-n-1$,
	\begin{align}
	\E_{B}[a(T-m)\vert\bfK,N(\bfK)=n]=& 1+\sum_{i=1}^{T-n-m}(1-\mu^*)^{i}+(1-\mu^*)^{T-n-m}\sum_{j=1}^{\infty}(1-\mu_{\text{min}})^{j}\nn\\
	=&\frac{1-(1-\mu^*)^{T-n-m+1}}{\mu^*}+\frac{(1-\mu^*)^{T-n-m}}{\mu_{\text{min}}},\label{eq:16}
	\end{align}
	and for $T-n \leq m \leq T-1$,
	\begin{align}
	\E_{B}[a(T-m)\vert\bfK,N(\bfK)=n]\leq \frac{1}{\mu_{\text{min}}}.\label{eq:17}
	\end{align}
	Combining \eqref{eq:15}, \eqref{eq:16} and \eqref{eq:17},
	\begin{align}
	&\E\left[\sum_{m=0}^{T-1}\E_{B}[a(T-m)\vert\bfK,N(\bfK)=n]\right]\nn\\ &\hspace{0.15in}\leq\E\left[\sum_{m=0}^{T-n-1}\left(\frac{1-(1-\mu^*)^{T-n-m+1}}{\mu^*}+\frac{(1-\mu^*)^{T-n-m}}{\mu_{\text{min}}}\right)
	+\sum_{m=T-n}^{T-1}\frac{1}{\mu_{\text{min}}}\right]\hspace{1in}\nn\\
	&\hspace{0.15in}=\E\left[\sum_{m=0}^{T-n-1}\frac{1-(1-\mu^*)^{T-n-m+1}}{\mu^*}+\frac{(1-\mu^*)\left(1-(1-\mu^*)^{T-n}\right)}{\mu_{\text{min}}\mu^*}+\frac{n}{\mu_{\text{min}}}\right]\nn\\
	&\hspace{0.15in}\leq\E\left[\frac{T-n}{\mu^*}\left(1-(1-\mu^*)^{T-n+1}\right)+\frac{(1-\mu^*)\left(1-(1-\mu^*)^{T-n}\right)}{\mu_{\text{min}}\mu^*}+\frac{n}{\mu_{\text{min}}}\right]\nn\\
	&\hspace{0.15in}\leq\E\left[\frac{T-n}{\mu^*}+\frac{1-\mu^*}{\mu^*\mu_{\text{min}}}+\frac{n}{\mu_{\text{min}}}\right]=\frac{T}{\mu^*}+\frac{1-\mu^*}{\mu^*\mu_{\text{min}}}+\left(\frac{1}{\mu_{\text{min}}}-\frac{1}{\mu^*}\right)\E[N(\bfK)].\nn
	\end{align}
\end{proof}

The next lemma summarizes the results from Theorem 1 in \cite{auer2002finite} and {\color{black}Theorem 2} in \cite{kaufmann2012thompson} to provide upper bounds on the number of time-slots in which a sub-optimal channel is picked by UCB and Thompson Sampling. 
\begin{lemma}
	\label{lemma:UCBbound}
	Let $k(t)$ denote the index of the communication channel used in time-slot $t$ and $k^*$ be the index of the optimal channel. Let $\displaystyle \E_\text{UCB} \left[N(\mathbf{K}(T))\right]$ and $\displaystyle \E_{\text{TS}}
	\left[N(\mathbf{K}(T))\right]$ denote the expected number of time-slots in which a sub-optimal channel is picked in time-slots 1 to $T$ by UCB and Thompson {\color{black}Sampling} respectively. 
	Then, for $t>K$,
	\begin{align*}
	\E_\text{UCB} \left[N(\mathbf{K}(T))\right] &\leq (K-1)\left(\frac{32 \log T}{\Delta^{2}} + 1 + \frac{\pi^{2}}{3} \right), \\
	\E_\text{TS} \left[N(\mathbf{K}(T))\right] &\leq \OO(K  \log T){\color{black},}
	\end{align*}
	where $\Delta = \mu^*-\max_{k\neq k^*}\mu_{k}$.
\end{lemma}
We now use Lemmas \ref{lemma:expectedAgeBound2} and \ref{lemma:UCBbound} to prove Theorems \ref{theo:UB_AoI_regret_UCB_impro} and \ref{theo:UB_AoI_regret_TS_impro}.

\begin{proof}(Proof of Theorems \ref{theo:UB_AoI_regret_UCB_impro} and \ref{theo:UB_AoI_regret_TS_impro})\\
	Note that by Assumption \ref{assumption:initialConditions},
	$$\sum_{t=1}^T \E[a^*(t)]= \frac{T}{\mu^*}.$$
	From Lemma \ref{lemma:expectedAgeBound2}, we have that,
	\begin{align*}
	\sum_{t = 1}^{T} \E[a(t)]  \leq \frac{T}{\mu^*}+\frac{1-\mu^*}{\mu^*\mu_{\text{min}}}+\left(\frac{1}{\mu_{\text{min}}}-\frac{1}{\mu^*}\right)\E[N(\bfK).
	\end{align*}
	The results then follow by Lemma \ref{lemma:UCBbound}.
\end{proof}

\subsection{Proof of Theorems \ref{theo:UB_AoI_regret_Q_UCB} and \ref{theo:UB_AoI_regret_Q_THS}}

We use the following lemmas to prove Theorems \ref{theo:UB_AoI_regret_Q_UCB} and \ref{theo:UB_AoI_regret_Q_THS}. 

\begin{lemma}
	\label{lemma:expectedAgeBound_forQ}
	Let $k(t)$ denote the index of the communication channel used in time-slot $t$ and $k^*$ be the index of the optimal channel. Let $\mathbf{K}(T) = \{k(1), k(2), \cdots, k(T)\}$ be the sequence of channels used in time-slots $1$ to $T$ and $E_t$ be the event that $k(\tau) = k^*$ for $t-c \log T + 1 \leq \tau \leq t$. 	
	Then, for $c = \frac{-1}{\log(1-\mu^*)}$,
	\begin{align*}
	\sum_{t = 1}^{T} \E[a(t)]  \leq \frac{T}{\mu^{*}} + \frac{c \log T+1}{\mu_{\text{min}}}  + \frac{1}{\mu_{\text{min}}} \E \left[\sum_{t = c \log T + 1}^{T} \mathbbm{1}_{E_{t}^c }\right].
	\end{align*}
\end{lemma}
\begin{proof} By definition,
	$
	\PP(a(t) > \tau) = \prod_{i=0}^{\tau} \left(1-\mu_{k(t-i)}\right).
	$
	Note that since $a(t) \geq 1$ for all $t$,
	$
	\E[a(t)] = \sum_{\tau=0}^{\infty} \PP(a(t) > \tau).
	$
	It follows that,
	\begin{align}
	\E[a(t)]= \E[\E[a(t)]] = \E\left[\sum_{\tau=0}^{\infty} \PP(a(t) > \tau)\right] \E\left[ \sum_{\tau=0}^{\infty}\prod_{i=0}^{\tau} \left(1-\mu_{k(t-i)}\right)\right].
	\label{eq:doubleExpectation}
	\end{align}
	For $t \geq c \log T$, we define $E_t$ as the event that $k(\tau) = k^*$ for $t-c \log T + 1 \leq \tau \leq t$. Then,
	\begin{align*}
	\E\left[ \sum_{\tau=0}^{\infty}\prod_{i=0}^{\tau} \left(1-\mu_{k(t-i)}\right)\bigg | E_t\right]
	\leq \sum_{i=1}^{c \log T}\prod_{j=0}^{i}(1-\mu^{*})+ \sum_{i=c \log T+1}^{\infty}
	(1-\mu^{*})^{c \log T}\prod_{j=c \log T+1}^{i}(1-\mu_{\text{min}}).
	\end{align*}
	Note that,
	\begin{align*}
	&\hspace{0.14in} \sum_{i=1}^{c \log T}\prod_{j=0}^{i}(1-\mu^{*}) \leq \sum_{i=1}^{\infty}\prod_{j=0}^{i}(1-\mu^{*}) = \frac{1}{\mu^{*}}, \\
	\text{and } & \sum_{i=c \log T+1}^{\infty}(1-\mu^{*})^{c \log T}\prod_{j=c \log T+1}^{i}(1-\mu_{\text{min}})  \leq (1-\mu^{*})^{c \log T}\frac{1}{\mu_{\text{min}}} = \frac{1}{\mu_{\text{min}}T}.
	\end{align*}
	It follows that
	\begin{align}
	\label{eq:age_in_the_good_case}
	&\E\left[ \sum_{\tau=0}^{\infty}\prod_{i=0}^{\tau} \left(1-\mu_{k(t-i)}\right)\bigg | E_t\right] \leq \frac{1}{\mu^{*}} +\frac{1}{\mu_{\text{min}}T}.
	\end{align}
	Moreover, since $\mu_{k(t)} \geq \mu_{\min}$, for all $t$, 
	\begin{align}
	\label{eq:age_in_the_bad_case}
	&\E\left[ \sum_{\tau=0}^{\infty}\prod_{i=0}^{\tau} \left(1-\mu_{k(t-i)}\right)\bigg | E_t^{c}\right] \leq \frac{1}{\mu_{\text{min}}}.
	\end{align}	
	\noindent Note that
	\begin{align}
	E_{t}^c = \bigcup_{\tau=t - c \log T + 1}^{t} \{k(\tau) \neq k^*\}, \ 
	\mathbbm{1}_{E_{t}^c } \leq \sum_{\tau=t - c \log T + 1}^{t} \mathbbm{1}_{k(\tau) \neq k^*} 
	\label{eq:expectation_indicator}
	\end{align}	
	From \eqref{eq:doubleExpectation}, \eqref{eq:age_in_the_good_case}, \eqref{eq:age_in_the_bad_case}, and \eqref{eq:expectation_indicator}, 
	\begin{align}
	\label{eq:usefulForQ}
	\sum_{t = 1}^{T} \E[a(t)]&= \sum_{t = 1}^{c \log T} \E[a(t)] + \sum_{t = c \log T + 1 }^{T} \E[a(t)]\nn\\
	&\leq \frac{c \log T}{\mu_{\text{min}}} + \frac{T-c \log T}{\mu^{*}} +\frac{T-c \log T}{\mu_{\text{min}}T}+\frac{1}{\mu_{\text{min}}} \E \left[\sum_{t = c \log T + 1}^{T} \mathbbm{1}_{E_{t}^c } \right].
	\end{align}
\end{proof}

\begin{lemma}
	\label{lemma:boundForQPolicies}
	Let $E_t$ be the event that $k(\tau) = k^*$ for $t-c \log T + 1 \leq \tau \leq t$. Let $\E_{\text{Q-UCB}}[ \ ]$ and $\E_{\text{Q-TS}}[ \ ]$ denote expectation under the Q-UCB and Q-TS policies. Then, 
\begin{align*}
\E_{\text{Q-UCB}}\left[\sum_{t = c \log T + 1}^T \mathbbm{1}_{E_{t}^c }\right] &\leq cK \log^4 T + O\left(\frac{K}{T^2} \right), \\
\E_{\text{Q-TS}}\left[\sum_{t = c \log T + 1}^T \mathbbm{1}_{E_{t}^c }\right] &\leq cK \log^4 T + O\left(\frac{K}{T^2} \right).
\end{align*}
\end{lemma}
\begin{proof}
	Let $E_t^{(1)}$ be the event that $Ex(\tau) = 1$ for some $\tau \in t-c \log T + 1$ to $t$ and $E_t^{(2)}$ be the event that $Ex(\tau) = 0$ for $t-c \log T + 1 \leq \tau \leq t$ and $k(\tau) \neq k^*$ for some $\tau \in t-c \log T + 1$ to $t$. It follows that
	\begin{align}
	\label{eq:Eequalse1pluse2}
	\sum_{t = c \log T + 1}^T \mathbbm{1}_{E_{t}^c } \leq \sum_{t = c \log T + 1}^T \mathbbm{1}_{E_{t}^{(1)} } + \sum_{t = c \log T + 1}^T \mathbbm{1}_{E_{t}^{(2)} }. 
	\end{align}	
	By the dicussion after Corollary 7 in the supplementary material for \cite{krishnasamy2016learning},
	\begin{align}
	\label{eq:e1UCB}
	\E_{\color{black}\text{Q-UCB}}\left[\sum_{t = c \log T + 1}^T \mathbbm{1}_{E_{t}^{(1)} }\right] &\leq cK \log^4 T, \\
	\label{eq:e1THS}
	\E_{\text{Q-TS}}\left[\sum_{t = c \log T + 1}^T \mathbbm{1}_{E_{t}^{(1)} }\right] &\leq cK \log^4 T. 
	\end{align}
	By Lemma 9 in the supplementary material for \cite{krishnasamy2016learning}, for $T$ large enough,
	\begin{align}
	\label{eq:e2UCB}
	\E_{\text{Q-UCB}}\left[ \sum_{t = c \log T + 1}^T \mathbbm{1}_{E_{t}^{(2)} } \right] &= O\left(\frac{K}{T^2} \right), \\
	\label{eq:e2THS} 
	\E_{\text{Q-TS}}\left[ \sum_{t = c \log T + 1}^T \mathbbm{1}_{E_{t}^{(2)} } \right] &= O\left(\frac{K}{T^2} \right).
	\end{align}
	The results follow from \eqref{eq:Eequalse1pluse2}, \eqref{eq:e1UCB}, \eqref{eq:e1THS} \eqref{eq:e2UCB} and \eqref{eq:e2THS}.
\end{proof}

\begin{proof}[Proof of Theorems \ref{theo:UB_AoI_regret_Q_UCB} and \ref{theo:UB_AoI_regret_Q_THS}]
	
	Recall that by Assumption \ref{assumption:initialConditions},
	$$\sum_{t=1}^T \E[a^*(t)]= \frac{T}{\mu^*}.$$
	The result then follows by Lemmas \ref{lemma:expectedAgeBound_forQ} and \ref{lemma:boundForQPolicies}.
\end{proof}

\section{Conclusions}
We consider a variant of MAB, called AoI bandits. We first characterize a lower bound on the regret achievable by any policy for AoI bandits. Next, we analyze the performance of popular policies, namely UCB and Thompson Sampling for our setting and prove that they are order-optimal for AoI bandits. In addition, we analyze the performance of two policies, namely, Q-UCB and Q-Thompson Sampling proposed in \cite{krishnasamy2016learning}. The commonality between these four policies is that they are AoI-agnostic, i.e., conditioned on the number of times each channel is used in the past and the number of successful communications on each channel, these policies make decisions independent of the current AoI. We then propose four AoI-aware policies, which also take the current value of AoI into account while making decisions. Via simulations, we observe that the AoI-aware policies outperform the AoI-agnostic policies. 
\bibliographystyle{IEEEtran}
\bibliography{IEEEabrv,ref}

\newpage
\appendix

\begin{algorithm}[h]
		\DontPrintSemicolon 
		\textbf{Initialise:} Set $\hat{\mu}_k=0$ to be the estimated success probability of Channel $k$, $T_k(0)=0$ $\forall$ $k\in[K]$.\;
		\While{$1\leq t \leq K$}{
			Schedule update on Channel $k(t)=t$\;
			Receive rewards $X_{k(t)}(t)\sim \text{Ber}(\mu_{k(t)})$\; $\hat{\mu}_{k(t)}=X_{k(t)}(t)$\;
			$T_{k(t)}(t)=1$\;
			$t=t+1$}
		\While{$t\geq K+1$}{
			$\alpha_{k}(t)=\hat{\mu}_k(t) T_k(t-1)+1$,\;
			$\beta_{k}(t)=(1-\hat{\mu}_k(t)) T_k(t-1)+1$,\;
			Let $\text{limit(t)}=\min\limits_{k\in [K]}\frac{\alpha_{k}(t)+\beta_{k}(t)}{\alpha_{k}(t)}$\;
			\uIf {$a(t-1)>\text{limit(t)}$}{
				\textit{Exploit:} Select channel with highest estimated success probability
			}
			\Else{
				\textit{Explore:}\;
				Schedule update on Channel $k(t)$ such that $$ k(t)=\arg \max_{k\in [K]}\hat{\mu}_k(t)+\sqrt{\frac{8\log t}{T_k(t-1)}}$$}
			Receive reward $X_{k(t)}(t)\sim \text{Ber}(\mu_{k(t)})$\;
			$\hat{\mu}_{k(t)}=(\hat{\mu}_{k(t)}\cdot T_{k(t)}(t-1)+X_{k(t)}(t))/(T_{k(t)}(t-1)+1)$\;
			$T_{k(t)}(t)=T_{k(t)}(t-1)+1$\;
			$t=t+1$}
		\caption{{\sc AoI-Aware Upper Confidence Bound (AA-UCB)}}
		\label{algo:UCB_Thr}
\end{algorithm}

\begin{algorithm}[h]
	\DontPrintSemicolon 
	\textbf{Initialise:} Set $\hat{\mu}_k=0$ to be the estimated success probability of Channel $k$, $T_k(0)=0$ $\forall$ $k\in[K]$.\;
	\While{$1\leq t \leq K$}{
		Schedule update on Channel $k(t)=t$\;
		Receive rewards $X_{k(t)}(t)\sim \text{Ber}(\mu_{k(t)})$\; $\hat{\mu}_{k(t)}=X_{k(t)}(t)$\;
		$T_{k(t)}(t)=1$\;
		$t=t+1$}
	\While{$t\geq K+1$}{let $E(t)\sim\text{Ber}\left(\min\left\{ 1,3K\frac{\log^2 t}{t}\right\}\right)$\;
		
		\uIf {$E(t)=1\;\&\&\; a(t)<Thr$}{
			\textit{Explore:} Schedule update on a channel chosen uniformly at random
		}
		\Else{
			\textit{Exploit:} Schedule update on channel $k(t)$ such that $$ k(t)=\arg \max_{k\in [K]}\hat{\mu}_k(t)+\sqrt{\frac{\log^2 t}{2T_k(t-1)}}$$
		}
		Receive reward $X_{k(t)}(t)\sim \text{Ber}(\mu_{k(t)})$\;
		$\hat{\mu}_{k(t)}=(\hat{\mu}_{k(t)}\cdot T_{k(t)}(t-1)+X_{k(t)}(t))/(T_{k(t)}(t-1)+1)$\;
		$T_{k(t)}(t)=T_{k(t)}(t-1)+1$\;
		$t=t+1$}
	\caption{{\sc AoI-Aware Q-Upper Confidence Bound}(AA Q-UCB)}
	\label{algo:Q_UCB_Thr}
\end{algorithm}


\begin{algorithm}[h]
	\DontPrintSemicolon 
	\textbf{Initialise:} Set $\hat{\mu}_k=0$ to be the estimated success probability of Channel $k$, $T_k(0)=0$ $\forall$ $k\in[K]$.\;
	\While{$t\geq 1$}{let $E(t)\sim\text{Ber}\left(\min\left\{ 1,3K\frac{\log^2 t}{t}\right\}\right)$\;
		
		\uIf {$E(t)=1\;\&\&\; a(t)< Thr$}{
			\textit{Explore:} Schedule a update on a channel chosen uniformly at random\;
		}
		\Else{
			\textit{Exploit:}\;
				$\alpha_{k}(t)=\hat{\mu}_k(t) T_k(t-1)+1$,\;
				$\beta_{k}(t)=(1-\hat{\mu}_k(t)) T_k(t-1)+1,$\;
				For each $k\in[K]$, pick a sample $\hat{\theta}_k(t)$ of distribution,$$\hat{\theta}_k(t)\sim \text{Beta}(\alpha_{k}(t),\beta_{k}(t)).$$
			Schedule update on a Channel $k(t)$ such that $$ k(t)=\arg \max_{k\in [K]}\hat{\theta}_k(t)$$
		}
		Receive reward $X_{k(t)}(t)\sim \text{Ber}(\mu_{k(t)})$\;
		$\hat{\mu}_{k(t)}=(\hat{\mu}_{k(t)}\cdot T_{k(t)}(t-1)+X_{k(t)}(t))/(T_{k(t)}(t-1)+1)$\;
		$T_{k(t)}(t)=T_{k(t)}(t-1)+1$\;
		$t=t+1$}
	\caption{{\sc AoI-Aware Q-Thompson Sampling }(AA Q-TS)}
	\label{algo:Q_Ths_Thr}
\end{algorithm}

\end{document}